\newif{\ifdraft}\drafttrue
\draftfalse

\RequirePackage{fix-cm}
\ifdraft
\documentclass[draft]{article}   
\usepackage{showkeys}
\else
\documentclass[]{article}   
\fi

\usepackage{graphicx}

\usepackage{microtype}
\usepackage{mathtools}
\usepackage[utf8]{inputenc}
\usepackage{amssymb,amsmath}
\usepackage{fancyhdr}
\usepackage{stmaryrd}
\usepackage{xspace}
\usepackage{tikz}
\usepackage{multicol}
\usepackage[hidelinks]{hyperref}
\usepackage[a4paper,left=33mm,right=33mm, top=3.3cm, bottom=4.5cm]{geometry}

\usepackage{url}
\newtheorem{lemma*}{Lemma}

\usepackage{amsthm}

\newtheorem{theorem}{Theorem}[section]

\newtheorem{lemma}[theorem]{Lemma}
\newtheorem{proposition}[theorem]{Proposition}
\newtheorem{corollary}[theorem]{Corollary}
\newtheorem{definition}[theorem]{Definition}

\theoremstyle{break}

\theoremstyle{definition}
\newtheorem{example}[theorem]{Example}
\newtheorem{remark}[theorem]{Remark}
\newtheorem{question}[theorem]{Question}

\usepackage{amsfonts}

\usepackage{tabularx,hhline,rotating,enumitem,xspace}

\usepackage{mathdots}

\usepackage{prettyref}

\renewcommand{\setminus}{\mysetminus}

\newenvironment{aw}{\noindent\color{magenta} AW : }{}

\newenvironment{AM}{\noindent\color{red} AM : }{}

\newrefformat{thm}{Theorem~\ref{#1}}
\newrefformat{lem}{Lem\-ma~\ref{#1}}
\newrefformat{def}{Definition~\ref{#1}}
\newrefformat{cor}{Corollary~\ref{#1}} 
\newrefformat{prop}{Proposition~\ref{#1}}
\newrefformat{sec}{Section~\ref{#1}}
\newrefformat{cap}{Chapter~\ref{#1}}
\newrefformat{bsp}{Example~\ref{#1}}
\newrefformat{rem}{Remark~\ref{#1}}
\newrefformat{fig}{Figure~\ref{#1}}
\newrefformat{eq}{(\ref{#1})}
\newrefformat{ex}{Example~\ref{#1}}
\newrefformat{tab}{Table~\ref{#1}}
\newrefformat{app}{Appendix~\ref{#1}}
\newrefformat{alg}{Algorithm~\ref{#1}}
\newrefformat{qu}{Question~\ref{#1}}

\newcommand{\mysetminusD}{\raisebox{.8pt}{\hbox{\tikz{\draw[line width=0.6pt,line cap=round] (3.5pt,0pt) -- (0,5.2pt);}}}}
\newcommand{\mysetminusT}{\mysetminusD}
\newcommand{\mysetminusS}{\raisebox{.5pt}{\hbox{\tikz{\draw[line width=0.45pt,line cap=round] (2.2pt,0) -- (0,3.8pt);}}}}
\newcommand{\mysetminusSS}{\raisebox{.35pt}{\hbox{\tikz{\draw[line width=0.4pt,line cap=round] (1.5pt,0) -- (0,2.8pt);}}}}

\newcommand{\mysetminus}{\mathbin{\mathchoice{\mysetminusD}{\mysetminusT}{\mysetminusS}{\mysetminusSS}}}

\newcommand{\supp}{\mathop\mathrm{supp}}

\newcommand{\set}[2]{\left\{#1\; \middle|\; #2\right\}}

\newcommand{\oneset}[1]{\left\{\mathinner{#1}\right\}}

\newcommand{\abs}[1]{\left|\mathinner{#1}\right|}
\newcommand{\Abs}[1]{\left\Vert\mathinner{#1}\right\Vert}
\newcommand{\fg}{f.\,g.\xspace}

\newcommand{\gen}[1]{\left< \mathinner{#1} \right>}
\newcommand{\genr}[2]{\left< \, \mathinner{#1}\; \middle|\;\mathinner{#2} \, \right>}

\newcommand{\N}{\ensuremath{\mathbb{N}}}
\newcommand{\Z}{\ensuremath{\mathbb{Z}}}

\newcommand{\NP}{\ensuremath{\mathsf{NP}}\xspace} %
\renewcommand{\L}{\ensuremath{\mathsf{LOGSPACE}}\xspace} %
\newcommand{\TC}{\ensuremath{\mathsf{TC}^0}\xspace}
\newcommand{\Tc}[1]{\ensuremath{\mathsf{TC}^{#1}}\xspace}
\newcommand{\Ac}[1]{\ensuremath{\mathsf{AC}^{#1}}\xspace}
\newcommand{\Nc}[1]{\ensuremath{\mathsf{NC}^{#1}}\xspace}

\newcommand{\NC}{\ensuremath{\mathsf{NC}}\xspace}
\newcommand{\FO}{\ensuremath{\mathsf{FO}}\xspace}
\renewcommand{\P}{\ensuremath{\mathsf{P}}\xspace}

\renewcommand{\phi}{\varphi}
\newcommand{\eps}{\varepsilon}

\newcommand{\Sig}{\Sigma}

\newcommand{\Del}{\Delta}

\newcommand{\Oh}{\mathcal{O}}

\newcommand{\cC}{\mathcal{C}}

\newcommand{\BS}[2]{\ensuremath{\mathrm{\bf{BS}}_{#1,#2}}\xspace}

\newcommand{\WP}{\ensuremath{\mathrm{WP}}\xspace}
\newcommand{\CP}{\ensuremath{\mathrm{CP}}\xspace}
\newcommand{\PP}{\ensuremath{\mathrm{PP}}\xspace}
\newcommand{\CSMP}{\ensuremath{\mathrm{CSGMP}}\xspace}
\newcommand{\CMMP}{\ensuremath{\mathrm{CSMMP}}\xspace}

\newcommand{\ord}{\ensuremath{\mathop{\mathrm{ord}}}\xspace}

\newcommand{\oi}[1]{{#1}^{-1}}

\newcommand{\ssnq}{\subsetneqq}

\newcommand{\wt}[1]{\widetilde{#1}}

\newcommand{\sse}{\subseteq}

\newcommand\ie{i.\,e., }

\newcommand\eg{e.\,g.\xspace}

\renewcommand{\epsilon}{\varepsilon}

\tikzset{
	ncbar angle/.initial=90,
	ncbar/.style={
		to path=(\tikztostart)
		-- ($(\tikztostart)!#1!\pgfkeysvalueof{/tikz/ncbar angle}:(\tikztotarget)$)
		-- ($(\tikztotarget)!($(\tikztostart)!#1!\pgfkeysvalueof{/tikz/ncbar angle}:(\tikztotarget)$)!\pgfkeysvalueof{/tikz/ncbar angle}:(\tikztostart)$)
		-- (\tikztotarget)
	},
	ncbar/.default=0.5cm,
}

\begin{document}

	\title{The Conjugacy Problem\\ in Free Solvable Groups and Wreath Products of Abelian Groups is in \TC}

	\author{Alexei Miasnikov$^1$~$\qquad$~Svetla Vassileva$^2$~$\qquad$~Armin Wei\ss$^3$ \\[4mm]
		\small $^1$Stevens Institute of Technology, Hoboken, NJ, USA\\
		\small $^2$Champlain College, St-Lambert, QC, Canada\\
		\small $^3$Universität Stuttgart, Germany
	}

	\maketitle

	\begin{abstract}
		We show that the conjugacy problem in a wreath product $A \wr B$ is uniform-\TC-Turing-reducible to the conjugacy problem in the factors $A$ and $B$ and the power problem in $B$. If $B$ is torsion free, the power problem for $B$ can be replaced by the slightly weaker cyclic submonoid membership problem for $B$. Moreover, if $A$ is abelian, the cyclic subgroup membership problem suffices, which itself is uniform-\Ac0-many-one-reducible to the conjugacy problem in $A \wr B$.

		Furthermore, under certain natural conditions, we give a uniform \TC Turing reduction from the power problem in  $A \wr B$ to the power problems of $A$ and $B$. Together with our first result, this yields a uniform \TC solution to the conjugacy problem in iterated wreath products of abelian groups~-- and, by the Magnus embedding, also in free solvable groups. 
		\smallskip
		

	\end{abstract}

	\section{Introduction}
	
	The conjugacy problem is one of Dehn's fundamental algorithmic problems in group theory \cite{Dehn11}. It asks on input of two group elements (given as words over a fixed set of generators) whether they are conjugate. The conjugacy problem can be seen as a generalization of the word problem, which on input of one word asks whether it represents the identity element of the group.
	In recent years the conjugacy problem gained an increasingly important role in non-commutative cryptography; see for example  \cite{CravenJ12,GrigorievS09,KoLCHKP00,SZ1,WangWCO11}. These applications use the fact that it is easy to create elements which are conjugate, but to check whether two given elements are conjugate might be difficult even if the word problem is easy. In fact, there are groups where the word problem is in polynomial time, but the conjugacy problem is undecidable \cite{Miller1}. 
	Moreover, there are natural classes, like polycyclic groups, which have a word problem in uniform \TC \cite{Robinson93phd}, but the conjugacy problem not even known to be in \NP. Another example for such a huge contrast is the Baumslag group, whose word problem is decidable in polynomial time, but the conjugacy problem is conjectured to be non-elementary \cite{DiekertMW14}.
	
	The class \TC is a very low complexity class consisting of those problems which can be recognized by a family of constant depth and polynomial size Boolean circuits which also may use majority gates. We only consider ($\mathsf{Dlogtime}$-)uniform \TC (and subsequently simply write \TC for uniform \TC). The word problem of abelian groups as well as integer arithmetic (iterated addition, multiplication, division) are problems in \TC. However, there are not many groups known to have conjugacy problem in \TC. Indeed, without the results of this paper, the Baumslag-Solitar groups \BS1q \cite{DiekertMW14} and nilpotent groups \cite{MyasnikovW17} are the only natural examples we are aware of. On the other hand, there is a wide range of groups having word problem in \TC: all polycyclic groups \cite{Robinson93phd} and, more generally, by a recent result all solvable linear groups \cite{KoenigL17}. Also iterated wreath products of abelian groups are known to have word problem in \TC \cite{KrebsLR07}.

	The study of the conjugacy problem in wreath products has quite a long history: in \cite{Matthews66} Matthews proved that a wreath product $A \wr B$ has decidable conjugacy problem if, and only if, both $A$ and $B$ have decidable conjugacy problem and $B$ has decidable \emph{cyclic subgroup membership problem} (note that in \cite{Matthews66} this is called \emph{power problem}). As a consequence, she obtained a solution to the conjugacy problem in free metabelian groups. Kargapolov and Remeslennikov generalized the result by establishing decidability of the conjugacy problem in free solvable groups of arbitrary degree \cite{KargapolovR66}.
	
	A few years later Remeslennikov and Sokolov \cite{RemSok} also generalized Matthews results to iterated wreath products by solving the cyclic subgroup membership problem in these groups. They also showed that the Magnus embedding \cite{Magnus39} of free solvable groups into iterated wreath products of abelian groups preserves conjugacy~-- thus, giving a new proof for decidability of the conjugacy problem in free solvable groups.

	Later, in \cite{MyasnikovRUV10} a polynomial time algorithm for the conjugacy problem in free solvable groups has been given and in \cite{Vassileva11} it is shown that  for iterated wreath products of abelian groups Matthews' criterion \cite{Matthews66} can be actually checked in polynomial time. In \cite{MiasnikovVW16} this has been further improved to \L.
	Recently, in \cite{GulSU17}, Matthews result has been generalized to a wider class of groups without giving precise complexity bounds~-- see the discussion in last section.

	In this work we use the same technique as in \cite{Matthews66,MiasnikovVW16,Vassileva11} in order to give a precise complexity version of Matthews result.
	Moreover, we extend the result of \cite{MiasnikovVW16,Vassileva11} also in several directions.
	As in \cite{MiasnikovVW16}, at some points we need a stronger hypothesis than in \cite{Matthews66} though: it is not sufficient to assume that the cyclic subgroup membership problem is decidable in \TC in order to reduce the conjugacy problem in a wreath product to the factors. Instead, we need the stronger power problem to be in \TC: on input of two group elements $b$ and $c$ compute an integer $k$ such that $b^k = c$. 
	More precisely, we establish the following results:
	\begin{itemize}
		\item The word problem of $A\wr B$ is uniform-\Ac0-Turing-reducible to the word problems of $A$ and $B$.
		
		\item There is a uniform \TC Turing reduction from the conjugacy problem in $A\wr B$ to the conjugacy problems in $A$ and $B$ together with the power problem in $B$. If $B$ is torsion-free, the power problem can be replaced by the cyclic submonoid membership problem; if $A$ is abelian, the power problem can be replaced by the cyclic subgroup membership problem.
		
		\item The cyclic subgroup membership problem in $B$ is \Ac0-reducible to the conjugacy problem in $A\wr B$ and, if $A$ is non-abelian, then also the cyclic submonoid membership problem in $B$ is \Ac0-reducible to the conjugacy problem in $A\wr B$ 
		
		\item Suppose the orders of torsion elements of $B$ are $\beta$-smooth for some $\beta \in \N$. Then, the power problem in $A\wr B$ is uniform-\TC-Turing-reducible to the power problems in $A$ and $B$.
		As a corollary we obtain that iterated wreath products of abelian groups have conjugacy problem in uniform \TC. Using the Magnus embedding \cite{Magnus39,RemSok}, also the conjugacy problem in free solvable groups is in uniform \TC.
	\end{itemize}
	Notice that images of group elements under the Magnus embedding can be computed in \TC (since any image under homomorphisms of finitely generated monoids can be computed in \TC \cite{LangeM98}).
	Thus, for free solvable groups as well as for iterated wreath products of abelian groups, our results nail down the complexity of conjugacy precisely. This is because the word problem in \Z~is already hard for \TC (and so the conjugacy problem in free solvable groups is \TC-complete). Also for wreath products $A \wr B$ with $A$ abelian or $B$ torsion-free, we have a tight complexity bound because in this case there is a reduction from the cyclic subgroup membership problem (resp.\ cyclic submonoid membership problem) in $B$ to the conjugacy problem in $A\wr B$.
	
	To solve the conjugacy problem, we first deal with the word problem. For a free solvable group of degree $d$, we obtain a circuit of majority depth $d$. It is not clear how a circuit of smaller majority depth could be constructed.
	On the other hand, \cite{MyasnikovRUV10} presents an algorithm for the word problem running in cubic time for arbitrary solvability degree. This gives rise to the question whether the depth  (or the size) of circuits for the word and conjugacy problem of free solvable groups could be bounded uniformly independent of the degree. Note that a negative answer to this question would imply that $\TC \neq \Nc1$.

	We want to emphasize that throughout we assume that the groups are finitely generated. As wreath products we consider only \emph{restricted} wreath products, that is the underlying functions are required to have finite support.
	
	\paragraph{Outline.}
	\prettyref{sec:prelims} introduces some notation and recalls some basic facts on complexity. Then in \prettyref{sec:wreath}, we define wreath products and discuss the solution to the word problem. \prettyref{sec:cpwreath} and \prettyref{sec:itwreath}, the main parts, examine the conjugacy problem in wreath products resp.\ iterated wreath products. In order to do so, we deal with the power problem in iterated wreath products in \prettyref{sec:itwreath}. Finally, in \prettyref{sec:conclusion}, we discuss some open problems.
	This work is an extended version of the conference paper \cite{MiasnikovVW17}. It contains all proofs, some more examples and a slightly stronger version of \prettyref{thm:cpwreath}.

	\section{Preliminaries}\label{sec:prelims}
	
	\paragraph{Words.} An \emph{alphabet} is a (finite or infinite) set $\Sig$; an element $a \in \Sig$ is called a \emph{letter}. 
	The free monoid over $\Sig$ is denoted by $\Sig^*$; its elements are finite sequences of letters and they
	are called {\em words}. The multiplication of the monoid is concatenation of words. The identity element is the empty word $\epsilon$.

	\paragraph{Groups.}
	We consider a finitely generated (\fg) group $G$ together with a surjective homomorphism $\eta:\Sig^* \to G$ (a \emph{monoid presentation}) for some finite alphabet $\Sig$. Throughout, all groups we consider are finitely generated even if not mentioned explicitly.
	In order to keep notation simple, we suppress the homomorphism $\eta$ and consider words also as group elements. We write $w=_{G}w'$ as a shorthand for $\eta(w)=\eta(w')$ and $ w \in_G A$ instead of $\eta(w) \in \eta(A)$ for $A \sse \Sig^*$ and $w \in \Sig^*$. 
	Whenever it is clear that we deal with group elements $g,h \in G$, we simply write $g=h$ for equality in $G$. We always assume that $\Sig =_G \Sig^{-1}$.
	
	We say two group elements $g,h\in G$ are conjugate, and we write $g\sim h$, if there exists an element $x\in G$ such that $g^x = x^{-1}gx = h$. Similarly, we say two words $u$ and $v$ in generators of $G$ are conjugate, and we write $u\sim_G v$, if the elements of $G$ represented by $u$ and $v$ are conjugate as elements of $G$. 
	We denote by $\ord(g)$ the order of a group element $g$ (\ie the smallest positive integer $d$ such that $g^d=1$, or $\infty$ if no such integer exists). For $g \in G$, the cyclic subgroup generated by $g$ is denoted by $\gen{g}$.
	A $d$-fold commutator is a group element of the form $x^{-1} y^{-1} xy$ for $(d-1)$-fold commutators $x$ and $y$; a $0$-fold commutator is any group element. The \emph{free solvable group} of \emph{degree} $d$ is the group subject only to the relations that all $d$-fold commutators are trivial. 
	
	\subsection{Complexity}
	Computation or decision problems are given by functions $f:\Del^* \to \Sig^*$ for some finite alphabets $\Del $ and $\Sig$. A decision problem (or formal language) $L$ is identified with its characteristic function $\chi_L: \Del^* \to \oneset{0,1}$ with $\chi_L(x)=1$ if, and only if, $x \in L$.

	\paragraph{Computational Problems in Group Theory.}
	
	Let $G$ be a group with finite generating set $\Sigma$. We define the following algorithmic problems in group theory.
	
	\begin{itemize}
		
		\item The \emph{word problem}  $ \WP(G) $  of $G$, is the set of all words representing the identity in $G$.
		
		\item The \emph{conjugacy problem} $\CP(G)$ is the set of all pairs $(v,w)$ such that $v \sim_G w$.
		
		\item The \emph{cyclic subgroup membership problem} $\CSMP(G)$: the set of all pairs $(v,w)$ such that $w \in \gen{v}$ (\ie there is some $k \in \Z$ with $v^k=_Gw$).
		
		\item The \emph{cyclic submonoid membership problem} $\CMMP(G)$: the set of all pairs $(v,w)$ such that  $w \in_G \oneset{v}^*$ (\ie there is some $k \in \N$ with $v^k=_Gw$).
		
		\item 
		The \emph{power problem} $\PP(G)$: on input of some $(v,w)\in \Sig^*\times \Sig^*$ decide whether there is some $k\in \Z$ such that $v^k=_G w$ and, in the ``yes'' case, compute this $k$ in binary representation. If $v$ has finite order in $G$, the computed $k$ has to be the smallest non-negative such $k$. 
	\end{itemize}
	Whereas the first four of these problems are decision problems, the last one is an actual computation problem.
	Be aware that sometimes in literature the power problem is defined as what we refer to as cyclic subgroup membership problem.

	\paragraph{Circuit Classes.}
	The class \Ac{0} is defined as the class of functions computed by families of circuits of constant depth and polynomial size with unbounded fan-in Boolean gates (and, or, not). \TC additionally allows majority gates. A majority gate (denoted by $\mathrm{Maj}$) returns $1$ if the number of $1$s in its input is greater or equal to the number of $0$s. In the following, we always assume that the alphabets $\Del$ and $\Sig$ are encoded over the binary alphabet $\oneset{0,1}$ such that each letter uses the same number of bits. Moreover, we assume that also the empty word $\epsilon$ has such a encoding over $\oneset{0,1}$, which is denoted by $\epsilon$ as well (be aware of the slight ambiguity). The empty word letter is used to pad outputs of circuits to fit the full number of output bits; still we do not forbid to use it in the middle. We say a function $f$ is \emph{\Ac0-computable} (resp.\ \emph{\TC-computable}) if $f \in \Ac0$ (resp.\ $f \in \TC$).

	In the following, we only consider $\mathsf{Dlogtime}$-uniform circuit families. $\mathsf{Dlogtime}$-uniform means that there is a deterministic Turing machine which decides in time $\Oh(\log n)$ on input of two gate numbers (given in binary) and the string $1^n$ whether there is a wire between the two gates in the $n$-input circuit and also decides of which type some gates is. Note that the binary encoding of the gate numbers requires only $\Oh(\log n)$ bits~-- thus, the Turing machine is allowed to use time linear in the length of the encodings of the gates.
	For more details on these definitions we refer to \cite{Vollmer99}.  In order to keep notation simple we write $\Ac0$ (resp.\ $\TC$) for $\mathsf{Dlogtime}$-uniform $\Ac{0}$ (resp.\ $\mathsf{Dlogtime}$-uniform $\Tc0$) throughout.
	We have the following inclusions (note that even $\TC \sse \P$ is not known to be strict):
	\begin{align*}
	\Ac0 \ssnq \TC \sse \L \sse \P.
	\end{align*}
	The following facts are well-known and will be used in the following without further reference:
	\begin{itemize}
		\item Barrington, Immerman, and Straubing \cite{BarringtonIS90} showed that $\TC = \FO(+,*,\mathrm{Maj})$, \ie  \TC comprises exactly those languages which are defined by some first order formula with majority quantifiers where positions may be compared using $+$, $*$ and $<$. In particular, if we can give a formula with majority quantifiers using only addition and multiplication predicates, we do not need to worry about uniformity.
		\item Homomorphisms can be computed in \TC \cite{LangeM98}: on input of two alphabets $\Sig$ and $\Del$ (coded over the binary alphabet), a list of pairs $(a,v_a)$ with $a \in \Sig$ and $v_a \in \Del^*$ such that each $a \in \Sig$ occurs in precisely one pair, and a word $w \in \Sig^*$, the image $\phi(w)$ under the homomorphism $\phi$ defined by $\phi(a) = v_a$ can be computed in \TC. Moreover, if $\phi$ is length-multiplying (that is $\phi(a)$ and $\phi(b)$ have the same length for all $a,b \in \Sig$),  the computation is in \Ac0. Note that by padding with the empty-word letter $\epsilon$, we can assume that all homomorphisms are length-multiplying.
		
		\item Iterated addition is the following problem: given $n$ numbers $a_1, \dots, a_n$ (in binary), compute $\sum_{i=1}^{n}a_i$ (as binary number).
		This is well-known to be in \TC.
	\end{itemize}
	
	\begin{example}\label{ex:abelian}
		Finitely generated abelian groups have word problem in \TC: the word problem of $\Z$ is in \TC using iterated addition (summing up numbers $1$ and $-1$), the word problem of finite cyclic groups is in \TC by iterated addition and then calculating modulo; and, finally, a word in a direct product is the identity if, and only if, it is the identity in all components. 
	\end{example}
	
	\begin{example}\label{ex:sortTC}
		Let $(k_1, v_1), \dots, (k_n,v_n)$ be a list of $n$ key-value pairs $(k_i, v_i)$ equipped with a total order on the keys $k_i$ such that it can be decided in \TC whether $k_i < k_j$. We assume that all pairs $(k_i, v_i)$ are encoded with the same number of bits. 
		It is a standard fact that the problem of sorting the list according to the keys is in \TC (\ie the desired output is a list $(k_{\pi(1)}, v_{\pi(1)}), \dots, (k_{\pi(n)},v_{\pi(m)})$ for some permutation $\pi$ such that $k_{\pi(i)} < k_{\pi(j)}$ for all $i < j$).
		
		We briefly describe a circuit family to do so: The first layer compares all pairs of keys $k_i,k_j$ in parallel. The next layer  for all $i$ and $j$ computes a predicate $P(i,j)$ which is true if, and only if, $\abs{\set{\ell}{k_\ell < k_i}} = j$. The latter is computed by iterated addition. 
		As a final step the $j$-th output pair is set to $(k_i, v_i)$ if, and only if, $P(i,j)$ is true.
	\end{example}
	
	\paragraph{Reductions.} Let $K\sse \Del^*$ and $L \sse \Sig^*$ be languages and $\cC$ a complexity class. Then $K$ is called \emph{$\cC$-many-one-reducible} to $L$ if there is a $\cC$-computable function $f: \Del^* \to \Sig^*$ such that $w \in K$ if, and only if, $f(w) \in L$. In this case, we write $K \leq_{\mathrm{m}}^\cC L$.
	
	A function $f$ is \emph{\Ac0-(Turing)-reducible} to a function $g$ if there is a $\mathsf{Dlogtime}$-uniform family of \Ac0 circuits computing $f$ which, in addition to the Boolean gates, also may use oracle gates for $g$ (\ie gates which on input $x$ output $g(x)$). This is expressed by $f \in \Ac{0}(g)$ or $f \leq_T^{\Ac0} g$. 
	%
	Likewise \TC (Turing) reducibility is defined. 
	Note that if $L_1, \dots, L_k$ are in \TC, then $\TC(L_1, \dots, L_k) = \TC$ (see \eg \cite{Vollmer99}).

	\begin{remark}
		The cyclic subgroup membership problem, in particular, allows to solve the word problem: some group element is in the cyclic subgroup generated by the identity if, and only if, it is the identity. Moreover, the cyclic subgroup membership problem for $(v,w)$ can be decided by two calls to the cyclic submonoid membership problem (for $(v,w)$ as well as for $(v^{-1},w)$).
		Also, the power problem is a stronger version of the cyclic submonoid membership problem (simply check the sign of the output of the power problem).
		Thus, we have
		\[ \WP(G) \leq_{\mathrm{m}}^{\Ac0}\CSMP(G) \leq_T^{\Ac0}\CMMP(G) \leq_T^{\Ac0}\PP(G).\]
	\end{remark}
	
	Moreover, the power problem enables to decide whether an element is of finite order (just compute the $k$ such that $g^k =_G g^{-1}$~-- if this is a positive number, then $g$ is of finite order, otherwise not).

	\begin{example}\label{ex:bs}
		Let $\BS12= \genr{a,t}{tat^{-1} = a^2}$ be the Baumslag-Solitar group. The conjugacy problem of $\BS12$ is in \TC by \cite{DiekertMW14}. Moreover, let us show that the power problem is also in \TC:  $\BS12$ is the semi-direct product $\Z[1/2] \rtimes \Z$ with multiplication defined by $(r,m)\cdot(s,q) = (r + 2^m s,m+q)$~-- see \eg\ \cite{DiekertMW14}. Any word of length $n$ over the generators can be transformed in \TC to a pair $(r,m)$ with $m\leq n$ and $r$ can be written down with $\Oh(n)$ bits in binary. Let $(r,m)$ and $(s,q)$ be two such inputs for the power problem. We wish to decide whether there is some $\ell$ with $(r,m)^\ell =(s,q)$: if $q \neq 0$, then the only possibility for $\ell$ is $\ell = q/m$.  If this is an not integer, then there is no such $\ell$. If it is, one needs to check whether it satisfies $(r,m)^\ell =(s,q)$. Because $\ell$ is bounded by the input length, this can be done in \TC using the circuit for the word problem \cite{Robinson93phd,DiekertMW14}. 
		
		Now let $q=0$. If also $s=0$, then the solution is $\ell = 0$. So let $s\neq 0$. If $m \neq 0$, clearly there is no solution, so we are in the case $q=m=0$ and $r,s \neq 0$. But now, again we simply need to compute $\ell = s/r$ (this can be done in \TC using Hesse's circuit for division \cite{hesse01,HeAlBa02}). If it is an integer, the power problem has the solution $\ell$, otherwise, it does not have a solution.

		Notice that this example shows that there are natural groups where the power problem can be solved in \TC, but~-- because of the exponential distortion of the subgroup $\gen{a}$~-- the solution to the power problem can only be returned if encoded in binary.
	\end{example}

	\section{Wreath Products and the Word Problem}\label{sec:wreath}
	Let $A$ and $B$ be groups. For a function $f: B \to A$ the \emph{support} of $f$ is defined as $\supp(f) = \set{b\in B}{f(b) \neq 1}$.
	For two groups $A$ and $B$, the set of functions from $B$ to $A$ with finite support is denoted by $A^{(B)}$; it forms a group under point-wise multiplication. 
	Mapping $a\in A$ to the function 
	\begin{align}
	a(b) = \begin{cases} a &\mbox{  if $b=1$,} \\ 1 &\mbox{  otherwise,}\end{cases}\label{eq:afunc}
	\end{align}
	gives an embedding of $A$ into $A^{(B)}$. 
	In what follows we identify $A$ with its image in $A^{(B)}$.
	The \emph{wreath product} $A\wr B$ of $A$ and $B$ is defined as the semi-direct product $B\ltimes A^{(B)}$, where the action of $b\in B$ on a function $f\in A^{(B)}$ is defined by $f^b(x) = f(xb^{-1})$. Note that this is also referred to as restricted wreath product. We identify $B$ and $A^{(B)}$ (and hence also $A$) with their canonical images in  $A\wr B$. Thus, for the multiplication in $A \wr B$ we have the following rules
	\begin{align*}
	\qquad (b,f)(c,g) &= (bc, f^cg),& (b,f)^{-1} &= (b^{-1}, (f^{-1})^{b^{-1}})\end{align*}
	for $b,c \in B$ and $f,g \in A^{(B)}$, where $f^{-1}$ is the point-wise inverse (\ie $f^{-1}(b) = (f(b))^{-1}$ for all $b \in B$).
	
	Let $\Sig_A$ and $\Sig_B$ be fixed  generating sets of $A$ and $B$, correspondingly. Then, $A \wr B$ is generated by $\Sig = \Sig_A \cup \Sig_B$ (using the embedding \prettyref{eq:afunc} of $A$ into $A \wr B$). 
	Given a word $w \in \Sig^*$ of length $n$, we can group it as $w = a_1b_1 \cdots a_mb_m $ with $a_i \in \Sig_A^*$, $b_i \in \Sig_B^*$ and $m\leq n$. Introducing factors $b\oi b \in \Sig_B^*$, we can rewrite this as follows:
	\begin{eqnarray*}
		w &=_G& a_1b_1 \cdots a_mb_m =_G b_1\,b_1^{-1}a_1b_1\cdots a_mb_m =_G b_1\:a_1^{b_1}a_2b_2\cdots a_mb_m\\
		&=_G&  b_1b_2\,(a_1^{b_1}a_2)^{b_2} \cdots a_mb_m =_G b_1b_2\: a_1^{b_1b_2}a_2^{b_2} \cdots a_mb_m\\
		&=_G& b_1\cdots b_m \cdot a_1^{b_1\cdots b_m}\cdots a_m^{b_m}
	\end{eqnarray*}
	Thus, we have $w =_G (b,f)$ with $b=b_1\cdots b_m$ and $f=a_1^{b_1\cdots b_m}\cdots a_m^{b_m}$. 
	Since $a^c$ and $a'^{c'}$ commute for distinct $c,c' \in B$ and for any $a,a'\in A$, we can reorder this product 
	to ensure that the exponents are distinct:  whenever we have $b_i\cdots b_m =_B b_j\cdots b_m$ for $i<j$, we combine the corresponding terms into a single term $(a_ia_j)^{b_i\cdots b_m}$. Thus, we can rewrite $f$ as the product $\tilde a_1^{\tilde b_1}\ldots \tilde a_k^{\tilde b_k}$, where $\tilde a_1, \ldots, \tilde a_k \in \Sig_A^*$, and $\tilde b_1, \ldots, \tilde b_{k}\in \Sig_B^*$ all represent distinct elements of $B$. Moreover, we can assume that all $\tilde a_i$ represent non-trivial elements of $A$. With this notation, we have $f(\tilde b_i) = \tilde a_i \neq 1$ and $f(c)=1$ for $c \not\in \{\tilde{b}_1, \dots, \tilde b_k\} = \supp(f)$. Furthermore, $f$ is completely given by the set of pairs $ \{(\tilde b_1, \tilde a_1), \ldots, (\tilde b_k, \tilde a_k)\}$.

	In the following, we always assume that a function $f \in A^{(B)}$ is represented as a list of pairs $f=((\tilde b_1, \tilde a_1), \ldots, (\tilde b_k, \tilde a_k))$ with  $\{\tilde{b}_1, \dots, \tilde b_k\} = \supp(f)$. The order of the pairs does not matter -- but they are written down in some order. We also assume for an input $w$ of length $n$, that $k=m=n$ and that every word $\tilde b_i$, $\tilde a_i$ has length $n$. This is achieved by padding with pairs $(\eps,\eps)$ (where $\epsilon$ is the letter representing the empty word).
	
	\begin{lemma}\label{lem:reduce}
		Let $A$ and $B$ be finitely generated groups and let $G= A \wr B$. There is an $\Ac0(\WP(A),\WP(B))$ circuit family which on input $w \in \Sig^*$ computes $(b,f)$ with $w =_G (b,f)$ where $b \in \Sig_B^*$ and $f$ is encoded as described in the preceding paragraph.
	\end{lemma}
	\begin{proof}
		For an input word $w = w_1 \cdots w_n \in \Sig^*$, we first calculate the image under the projection $\pi_B: a \mapsto \epsilon$ for $a \in \Sig_A$. Since $\epsilon$ is a letter in our alphabet, this is a length-preserving homomorphism, and thus, can be computed in \Ac0 \cite{LangeM98}. We have $b = \pi_B(w)$. 
		Next, define the following equivalence relation $\approx$ on $\oneset{1, \dots, n}$:
		\begin{align*}
		i \approx j \iff  \pi_B(w_{i+1} \cdots w_{n})=_B \pi_B(w_{j+1} \cdots w_{n})
		\end{align*}
		After the computation of $\pi_B$ it can be checked for all pairs $i,j$ in parallel whether $i\approx j$ using $\binom{n}{2}$ oracle calls to the word problem of $B$. Let $[i]$ denote the equivalence class of $i$.
		Now, $\oi b w$ is in the (finite) direct product $\prod_{[i]} A^{\pi_B(w_{i+1} \cdots w_n)}  \leq A^{(B)}$ (this is well-defined by the definition of $\approx$). The projection to the component associated to $[i]$ is computed by replacing all $w_j$ by $\epsilon$ whenever $w_j \in \Sig_B$ or $j\not \approx i$. As  before, this computation is in \Ac0. As a representative of $[i]$, we choose the smallest $i \in [i]$. Now, the preliminary output is the pair $(b,(f_1, \dots, f_n))$ with
		$$f_i=\begin{cases}
		\left(\pi_B(w_{i+1} \cdots w_n), \prod_{j \in [i]} w_j\right) & \text{if } i = \min[i],\\
		(\eps,\eps) & \text{otherwise.}
		\end{cases}$$
		Up to the calculation of $\approx$, everything can be done in $\Ac0$ (checking $ i = \min[i]$ amounts to $\bigwedge_{j<i} \lnot(i \approx j)$). Finally, pairs $f_i= (b_i,a_i)$ with $a_i=_A 1$ are replaced by $(\eps,\eps)$. This requires an additional layer of calls to the word problem of $A$.
		
		If we assign appropriate gate numbers corresponding to the description of our circuit (\eg concatenation of the number of the layer and the indices $i$, $j$), it is easy to see that it can be checked in linear time on input of two binary gate numbers if the two gates are connected. This establishes uniformity of the circuit.
		\end{proof}

	\begin{theorem}\label{thm:wpwreath}
		$\WP(A \wr B) \in \Ac0(\WP(A),\WP(B))$.
	\end{theorem}
	\begin{proof}
		This is an immediate consequence of \prettyref{lem:reduce} since $ (b,f)=_G 1$ if, and only if, $b=_B1$ (can be checked using the word problem of $B$) and $f = ((\eps,\eps), \dots, (\eps,\eps))$.
		\end{proof}
	Note that \prettyref{thm:wpwreath} is a stronger version of \cite{Waack90} where \Nc1 reducibility is shown.
	
	\begin{definition} Let $d \in \N$. We define the \emph{left-iterated wreath product}, $A\mathop{^d\wr} B$, and the \emph{right-iterated wreath product} $A\wr^d B$ of two groups $A$ and $B$ inductively as follows:
		\begin{multicols}{2}
			\begin{itemize}
				\item $A\mathop{^1\wr} B = A\wr B$
				\item $A\mathop{^d\wr} B = A\wr (A\;{^{d-1}\wr}\; B)$
			\end{itemize}
			\begin{itemize}
				\item $A\wr^1 B = A\wr B$
				\item $A\wr^d B = (A\wr^{d-1} B)\wr B$
			\end{itemize}
		\end{multicols}	
	\end{definition}

	Let $S_{d,r}$ denote the free solvable group of degree $d$ and rank $r$. The Magnus embedding \cite{Magnus39} is an embedding $S_{d,r} \to \Z^r \wr S_{d-1, r}$. By iterating the construction, we obtain an embedding $S_{d,r} \to \Z^r\mathop{^{d}\wr} 1$.
	For the purpose of this paper, the explicit definition of the homomorphism is not relevant~-- it suffices to know that it is an embedding and that it preserves conjugacy \cite{RemSok}.
	The following corollary is also a consequence of \cite{KrebsLR07} since a wreath product can be embedded into the corresponding block product.

	\begin{corollary}
		Let $A$ and $B$ be f.\,g.\ abelian groups and let $d \geq 1$. The word problems of $A\wr^d B$ and of $A\mathop{^d\wr} B$ are in \TC. In particular, the word problem of a non-trivial free solvable group is \TC-complete.
	\end{corollary}
	Note that here the groups $A$, $B$ and the number $d$ of wreath  products are fixed. Indeed, if there were a single \TC circuit which worked for free solvable groups of arbitrary degree, this circuit would also solve the word problem of the free group, which is \Nc1-hard~-- thus, showing \TC= \Nc1.
	
	\begin{proof}
		The first statement follows from \prettyref{thm:wpwreath} because f.\,g.\ abelian groups have word problem in \TC (see \prettyref{ex:abelian}). The second statement then follows by the Magnus embedding \cite{Magnus39} and the fact that homomorphisms can be computed in \TC. The hardness-part is simply due to the fact that a non-trivial free solvable group has an element of infinite order, \ie a subgroup $\Z$, whose word problem is hard for \TC.
		\end{proof}

	\begin{remark}
		For a \TC circuit, the \emph{majority depth} is defined as the maximal number of majority gates on any path from an input to an output gate (see \eg \cite{MacielT98}). 
		Assume that $\WP(A), \WP(B) \in \TC$. The circuit in the proof of \prettyref{lem:reduce} contains one layer of oracle gates to the word problem of $B$ followed by a layer of oracle gates to the word problem of $A$. The additional check for $b=_B1$ in the proof of \prettyref{thm:wpwreath} can be done in parallel to the computation of \prettyref{lem:reduce}; thus, it can be viewed as part of the layer of oracle gates for $\WP(B)$. Since \prettyref{lem:reduce} is an \Ac0 reduction, the majority depth of the resulting circuit is at most $m_A + m_B$ where $m_A$ (resp.\ $m_B$) is the majority depth of the circuit family for $\WP(A)$ (resp. $\WP(B)$).
		
		Starting with the word problem of a free abelian group $\Z^r$, which is in \TC with majority depth one, we see inductively that a $d$-fold iterated wreath product $\Z^r\mathop{^{d}\wr} 1$~-- and thus the free solvable group of degree $d$~-- has word problem in \TC with majority depth at most $d$. On the other hand, we do not see a method how to improve this bound any further.
		In \cite{KrebsLR07} a similar observation was stated for iterated block products (into which wreath products can be embedded). There the question was raised how the depth of the circuit for the word problem (or more general any problem recognized by the block product) is related to the number of block products in an iterated block product (the so-called block-depth).
	\end{remark}
	
	\begin{question}\label{qu:majoritydepth}
		Can the word problem of a free solvable group of degree $d$ be decided in \TC with majority depth less than $d$?
	\end{question}
	We want to point out that \prettyref{qu:majoritydepth} is related to an important question in complexity theory: as outlined in \cite{MacielT98}, a negative answer would imply that $\TC \neq \NC^1$.
	Nevertheless, the following observations point rather towards a positive answer of \prettyref{qu:majoritydepth}:
	the word problem of free solvable groups is decidable in time $\Oh(n^3)$~-- regardless of the solvability degree $d$ \cite{MyasnikovRUV10,Vassileva11}.
	Moreover, the circuit for linear solvable groups (\emph{not} for free solvable groups with $d>2$) from \cite{KoenigL17} can be arranged with majority depth bounded uniformly for all groups. This is because every matrix entry in a product of upper triangular matrices can be obtained as iterated addition of iterated multiplications of the entries of the original matrices (for the precise formula, see \cite{KoenigL17}). These operations have circuits of uniformly bounded depth (also for f.\,g.\ field extensions). Hence, only the size of the circuits, but not the depth, depends on the solvability degree.

	\section{The Conjugacy Problem in Wreath Products}\label{sec:cpwreath}
	In order to give a \TC reduction of the conjugacy problem of $A\wr B$ to the conjugacy problems of $A$ and $B$ and the power problem of $B$, we follow Matthews' outline \cite{Matthews66}, where the same reduction was done for decidability.
	For deciding conjugacy of two elements $(b,f), (c,g)$ in a wreath product $A\wr B$ we will study the behavior of $f$ and $g$ on cosets of $\langle b \rangle \leq B$. 
	For $ b,d, t\in B$, $f\in A^{(B)}$, and $t \in T$, we define $$\pi_{t,b}^{(d)} (f) = \left\{ \begin{array}{cl} \prod\limits_{j=0}^{N-1} f(tb^jd^{-1}) & \;\text{ if } \ord(b)=N<\infty, \\ \prod\limits_{j=-\infty}^{\infty} f(tb^jd^{-1}) & \;\text{ if } \ord(b)=\infty, \end{array} \right.$$
	which is an element of $A$.
	We denote $\pi_{t,b}^{(1)}(f)$ by $\pi_{t,b}(f)$. 
	The definition of the $\pi_{t,b}$ depends on the order of $b$. However, observe that even in the case when the order of $b$ is infinite the product is finite since the function $f$ is of finite support. In fact, it is the product of all possible non-trivial factors of the form $f(tb^jd^{-1})$ multiplied in increasing order of $j$. The same is true in the case when the order of $b$ is finite. So in order to compute $\pi_{t,b}^{(d)}$, we need to find all the elements of the form $tb^jd^{-1}$ for which $f$ is non-trivial, arrange them in increasing order of $j$ and concatenate the respective $a_j$.

	\begin{lemma}\label{lem:piprods}
		The computation of $\pi_{t,b}^{(d)} (f)$ is in $\TC(\PP(B))$. More precisely, the input is $b,d,t \in \Sig_B^*$ and a function $f=(( b_1,  a_1), \ldots, ( b_n,  a_n))$, the output is $\pi_{t,b}^{(d)} (f)$ given as a word over $\Sig_A$.
		Moreover,
		\begin{itemize}
			\item  if $B$ is torsion-free, then it is in $\TC(\CMMP(B))$,
			\item if $A$ is abelian, then it is in $\TC(\CSMP(B))$.
		\end{itemize}
	\end{lemma}

	\begin{proof}
		One needs to check for all $j$ whether $t^{-1}b_jd \in \gen{b}$ (for $(b_j,a_j)\in f$) and if so, the respective power $k_j$ such that $t^{-1}b_jd = b^{k_j}$ has to be computed. For all $j$ this can be done in parallel using oracle gates for the power problem of $B$. 
		The next step is to sort the tuples $(b_j, a_j)$ with $t^{-1}b_jd \in \gen{b}$ according to their power $k_j$. This can be done in \TC as described in \prettyref{ex:sortTC}. The output $\pi_{t,b}^{(d)} (f)$ is the product (in the correct order) of the respective $a_j$.
		
		Now, let $B$ be torsion-free. The exponents $k_j$ with $t^{-1}b_jd = b^{k_j}$ are only needed in order to sort the pairs $(b_j, a_j)$. Thus, it suffices to decide for given $j$ and $j'$ whether $k_j \leq k_{j'}$ (where $k_{j'}$ is defined analogously to $k_j$). Since we assumed that $b$ has infinite order, we have  $k_j \leq k_{j'}$ if, and only if, $(t^{-1}b_jd)^{-1}t^{-1}b_{j'}d =_G b^{-k_j}b^{k_{j'}} \in_G \set{b^k}{k\in \N}$
		that is if, and only if, $(t^{-1}b_jd)^{-1}t^{-1}b_{j'}d$ is in the cyclic submonoid generated by $b$. Therefore, we can replace the power problem by the cyclic submonoid membership problem in the torsion-free case.
		
		Finally, let $A$ be abelian. In this case, the order of the factors of $\pi_{t,b}^{(d)} (f)$ does not matter; hence, there is no need for sorting the factors. For checking $t^{-1}b_jd \in \gen{b}$, the cyclic subgroup membership problem suffices.
		\end{proof}

	A \emph{full system of $\langle b\rangle$-coset representatives} is a set $ T\sse B$ of such that $ t\gen{b} \cap t'\gen{b}  = \emptyset$ for $t \neq t' \in T$ and $B = T\gen{b} $.
	In \cite{Matthews66}, Matthews provides the following criterion for testing whether two elements of a wreath product are conjugate.

	\begin{proposition}[{{\cite[{Prop.\ 3.5 and 3.6}]{Matthews66}}}]
		\label{prop:matthews}\label{prop:conjugacy criterion}
		Let $A$ and $B$ be groups. Two elements $x=(b,f)$ and $y=(c,g)$ in $A\wr B$ are conjugate if, and only if, there exists $d\in B$ such that
		\begin{itemize}
			\item $db=cd$ in $B$ and
			\item if $\ord(b)$ is finite, $\pi_{t,b}(f)$ is conjugate to $\pi_{t,b}^{(d)}(g)$ for all $t \in T$,
			\item if $\ord(b)$ is infinite, $\pi_{t,b}(f)$ is equal to $\pi_{t,b}^{(d)}(g)$ for all $t \in T$,
		\end{itemize}
		where $T$ is a full system of $\gen{b}$-coset representatives.
	\end{proposition}

	\begin{example}
		Let $G= \Z_2 \wr \Z$ be the \emph{Lamplighter group} and let $(b,f),(c,g) \in G$ with $c,b \in \Z$, $f,g \in \Z_2^{(\Z)}$. We can view $f$ and $g$ as finite subsets of $\Z$ (\ie we identify $f$ with $\supp(f)$). Now the point-wise addition in $\Z_2^{(\Z)}$ becomes the symmetric difference $\triangle$ of subsets and we obtain the multiplication rule  $(b,f)(c,g) = (b+ c,(f+c)\mathop{\triangle} g )$ where $f +c$ is defined as $\oneset{f_1 + c, \dots, f_n + c}$ for $f = \oneset{f_1, \dots, f_n}$. 
		Now, $T = \oneset{0, \dots, b-1}$ if $b \neq 0$ and $T = \Z$ if $b=0$. For $t \in T$ and $d \in \Z$ we have 
		$$\pi_{t,b}^{(d)}(f) = \abs{\set{f_i \in f}{f_i \equiv t -d \,\bmod b}}\bmod 2.$$

		\prettyref{prop:conjugacy criterion} tells us that $(b,f) \sim (c,g)$ if, and only if, $b=c$ and there is some $d \in \Z$ such that 
		$$\abs{\set{f_i \in f}{f_i \equiv t \,\bmod b}}\equiv \abs{\set{g_i \in g}{g_i \equiv t -d \,\bmod b}} \mod 2$$
		for all $t \in T$ (or equivalently for all $t \in \Z$).
		
		In particular, $(1,f) \sim (1,g)$ as soon as $\abs{f} \equiv \abs{g}\, \bmod 2$ and  $(0,f) \sim (0,g)$ if, and only if, there is some $x \in \Z$ with $f = g+x$.
	\end{example}

	In order to derive a criterion for conjugacy, which is more suitable for working in \TC or \L,  \cite{MiasnikovVW16} follows the outline of \cite{Matthews66}. For completeness, we will give a similar criterion in \prettyref{prop:conjugacyCrit} and we will show how it follows from \prettyref{prop:matthews}.
	
	\begin{lemma}\label{lem:rep_choice}
		Let $c,d,e, r,s\in B$ with $d\langle c\rangle= e \langle c\rangle$ and $r\langle c\rangle= s \langle c\rangle$. Then for every $g\in A^{(B)}$, we have $\pi_{r,c}^{(d)}(g) \sim \pi_{s,c}^{(e)}(g)$ and, if $c$ has infinite order, we have $\pi_{r,c}^{(d)}(g) = \pi_{s,c}^{(e)}(g)$.
	\end{lemma}
	
	\begin{proof}
		Since $d\langle c\rangle= e \langle c\rangle$ and $r\langle c\rangle=s \langle c\rangle$, there are integers $p,q$ for which $d=ec^{p}$ and $r=sc^{q}$; hence,
		$$\pi_{r,c}^{(d)}(g) = \prod\limits_{k} g(rc^kd^{-1})  = \prod\limits_k g(sc^{q}c^kc^{-p}e^{-1}) = \prod\limits_k g(s c^{k + q - p}e^{-1}).$$
		In the infinite order case, the last product in the above equation is equal to $\prod_k g(sc^ke^{-1}) = \pi_{s,c}^{(e)}(g)$, in the finite order case it is a cyclic permutation of the factors in the product $\prod_k g(sc^ke^{-1}) = \pi_{s,c}^{(e)}(g)$ and hence is conjugate to $\pi_{s,c}^{(e)}(g)$. 
		\end{proof}

	\begin{proposition}\label{prop:conjugacyCrit}
		Let $x=(b,f)$ and $y=(c,g)$ be two elements of $A\wr B$ with $\supp(f) = \oneset{b_1, \dots, b_n}$ and $\supp(g)=\{\beta_1, \ldots, \beta_m\}$. Let $ \wt T= \set{\beta_i\beta_j^{-1}b_k}{1\leq i,j \leq m, 1\leq k \leq n}$. If $b$ and $c$ are not conjugate in $B$, then $x$ and $y$ are not conjugate in $A\wr B$. Otherwise, we distinguish the following cases: 
		\begin{enumerate}
			\item Suppose $\pi_{t,b}(f)=1$ for all $t \in \supp(f)$. Then $x\sim y$ if, and only if, $\pi_{s,c}(g)=1$ for all $s \in \supp(g)$.
			\item Suppose there exists some $t \in \supp(f)$ such that $\pi_{t,b}(f)\neq 1$. Then $x\sim y$ if, and only if, there is some $d\in\oneset{\beta_1^{-1}t, \ldots, \beta_m^{-1}t}$ such that $db=cd$ and 
			\begin{enumerate}
				\item $\pi_{t'\!,b}(f) = \pi_{t'\!,b}^{(d)}(g)$ for all $t' \in\wt T$ if $\ord(b)=\infty$, or
				\item $\pi_{t'\!,b}(f) \sim \pi_{t'\!,b}^{(d)}(g)$ for all $t' \in \wt T$ if $\ord(b)$ is finite. 
			\end{enumerate}  
		\end{enumerate}
	\end{proposition}

	\begin{proof}
		We have to show that the conditions of \prettyref{prop:conjugacyCrit} imply the condition of \prettyref{prop:matthews}. The proof follows the one of \cite[Thm.\ B]{Matthews66}. Let $T$ be the full system of $\gen{b}$-coset representatives of \prettyref{prop:matthews}. 
		
		First, observe that by  \prettyref{lem:rep_choice} the condition of \prettyref{prop:matthews} is invariant under change of the system of representatives $T$. Moreover, we can add multiple representatives of one coset to $T$ (\ie we do not need to require that  $ t\gen{b} \cap t'\gen{b}  = \emptyset$ for $t \neq t' \in T$ ) as long as $T \gen{b}=B$, without changing the condition of \prettyref{prop:matthews}. Hence, we can assume that $\wt T \sse T$ and \begin{align}(T \setminus \wt T) \cap \wt T \gen{b} = \emptyset\label{eq:TwtT}.\end{align}
		Let us show that 
		\begin{align}
		\pi_{t,b}(f)= 1 &&\text{and}&& \pi_{t\!,b}^{(d)}(g) = 1  &&&\text{for } t \in T \setminus \wt T \text{ and }d \in\oneset{\beta_1^{-1}t, \ldots, \beta_m^{-1}t}.\label{eq:pi0}
		\end{align}
		Let $t \in T$. If $\pi_{t,b}(f)\neq 1$, then $tb^\ell \in \supp(f) \subseteq \wt T$ for some $\ell\in \Z$; hence, by \prettyref{eq:TwtT}, $t \in \wt T  $. 
		If $\pi_{t,b}^{(d)}(g) \neq 1$ for some $d\in\oneset{\beta_1^{-1}b_k, \ldots, \beta_m^{-1}b_k}$, then $tb^\ell d^{-1} \in \supp(g)$ for some $\ell\in \Z$. Therefore, there is some $i \in \oneset{1, \dots, m}$ with $\beta_i = tb^\ell d^{-1} =  tb^\ell b_k^{-1}\beta_j$. Hence, $tb^\ell = \beta_i\beta_j^{-1}b_k \in \wt T$ and, by \prettyref{eq:TwtT}, $t=\beta_i\beta_j^{-1}b_k$. This shows \prettyref{eq:pi0}.

		Now consider $x=(b,f)$ and $y=(c,g)$. If $b$ and $c$ are not conjugate, then $x$ and $y$ are certainly not conjugate. If they are, we consider the following two cases:
		\begin{enumerate}
			\item 
			Suppose $\pi_{t,b}(f) = 1$ for all $t \in \wt T$. By the same argument as for \prettyref{eq:pi0}, this is the case if, and only if, $\pi_{t,b}(f) = 1$ for all $t \in T$. Let $S$ be a full system of $\gen{c}$-coset representatives. By \prettyref{prop:conjugacy criterion}, $x\sim y$ if, and only if, there is some $d \in B$ such that $db=cd$ and $\pi_{t,b}^{(d)}(g) = 1$ for all $t \in T$. Now, 
			\begin{equation*}
			\pi_{t,b}^{(d)}(g) = \prod\limits_{j} g(tb^jd^{-1}) = \prod\limits_{j} g(td^{-1}c^j) = \pi_{td^{-1}\!,c}(g). 
			\end{equation*}
			For each $t\in T$, there is some $s \in S$ with $td^{-1} \in s \gen{c}$ and vice-versa. 
			By \prettyref{lem:rep_choice},  $\pi_{td^{-1}\!,c}(g) \sim \pi_{s,c}(g)$ (resp.\ $\pi_{td^{-1}\!,c}(g) = \pi_{s,c}(g)$), and it follows that 
			$$\pi_{td^{-1}\!, c} (g) =1 \text{ for all } t \in T \;\iff\; \pi_{s,c}(g) = 1 \text{ for all } s \in S.$$
			Thus, $x\sim y$ if, and only if,  there is some $d \in B$ with $db=cd$ and $\pi_{s,c}(g) =1$ for all $s \in S$. 
			
			Assume that $\pi_{s,c}(g) \neq 1$ for some $s \in S$. Then $sc^\ell= \beta_i \in \supp(g) $ for some $\ell \in \Z$ and so $\pi_{\beta_i,c}(g) \neq 1$. Thus, $\pi_{s,c}(g) =1$ for all $s \in S$ if, and only if, $\pi_{s,c}(g) =1$ for all $s \in \supp(g)$.
			
			\item 
			Now, suppose that $\pi_{t,b}(f) \neq 1$ for some $t \in \supp(f)$ and let $x$ and $y$ be conjugate. By \prettyref{prop:conjugacy criterion}, there is some $d\in B$ such that $db=cd$ and $\pi_{t,b}(f)$ is conjugate (resp.\ equal) to $\pi_{t,b}^{(d)} (g)$.
			For this $d$ we have $\pi_{t,b}^{(d)}(g) \neq 1$. In particular, there is some $l\in \mathbb{Z}$ with $g(tb^ld^{-1}) \neq 1$ and so $ tb^ld^{-1} = \beta_i \in \supp(g)$ for some $i \in \oneset{1,\dots, m}$. Hence, $d \in \set{ \beta_1^{-1}tb^l, \ldots, \beta_m^{-1}tb^l}{l \in\Z}$. 
			We can assume $l=0$ because, if $d = eb^l$, then $db=cd$ if, and only if, $eb=ce$ and, by \prettyref{lem:rep_choice}, for every $t' \in \wt T$ we have	
			$\pi_{t'\!,b}^{(d)}(g)= \pi_{t'\!,b}^{(e)}(g)$ (resp. $\pi_{t'\!,b}^{(d)}(g) \sim \pi_{t'\!,b}^{(e)}(g)$). Thus, for some $d\in\oneset{\beta_1^{-1}t, \ldots, \beta_m^{-1}t}$ with $db=cd$ we have $\pi_{t'\!,b}(f) = \pi_{t'\!,b}^{(d)}(g)$ (resp.  $\pi_{t'\!,b}(f) \sim \pi_{t'\!,b}^{(d)}(g)$) for all $t' \in \wt T$.
			
			The converse implication follows immediately from \prettyref{eq:pi0} and \prettyref{prop:conjugacy criterion}.
		\end{enumerate}\vspace{-6mm}
		\end{proof}

	\begin{theorem}\label{thm:cpwreath}
		Let $A$ and $B$ be arbitrary finitely generated groups. We have
		\begin{itemize}
			\item  $\CP(A \wr B) \in \TC(\CP(A),\CP(B),\PP(B))$,
			\item  $\CP(A \wr B) \in \TC(\CP(A),\CP(B),\CMMP(B))$ if $B$ is torsion-free,
			\item $\CP(A \wr B) \in \TC(\CP(A),\CP(B),\CSMP(B))$ if $A$ is abelian.
		\end{itemize}
	\end{theorem}
	
	\begin{proof}
		By \prettyref{lem:reduce}, we may assume that the input is given as two pairs $(b,f)$ and $(c,g)$. 
		As before we write $\supp(f) = \oneset{b_1, \dots, b_n}$ and $\supp(g)=\{\beta_1, \ldots, \beta_m\}$. 
		By \prettyref{lem:piprods}, we can assume that $\pi_{t'\!,b}(f)$, $\pi_{t'\!,b}^{(d)}(g)$, and $\pi_{s,c}(g)$ for $d\in\oneset{\beta_1^{-1}t, \ldots, \beta_m^{-1}t}$,  $s \in \supp(g)$, and $t,t' \in \wt T$ are part of the input.

		Now, let us describe an \Ac0-circuit with oracle calls to the word and conjugacy problems of $A$ and $B$ which evaluates the criterion of \prettyref{prop:conjugacyCrit}. If $A$ is non-abelian and $B$ has torsion it also uses oracle gates for $\PP(B)$. 
		
		First, one call to the conjugacy problem in $B$ is performed for determining whether $b$ and $c$ are conjugate.
		Then, in the next stage the two cases can be distinguished by at most $|\wt T|$ calls to the word problem of $A$.
		Now, case (i) is simply a conjunction of calls to the word problem of $A$.
		Case (ii) is a disjunction over all possible values for $d$; for each value of $d$ it is again a conjunction of one call to the word problem of $B$ and several calls to the word problem of $A$ (case (ii a)) or the conjugacy problem in $A$ (case (ii b)). Cases  (ii a) and  (ii b) can be distinguished using the power problem in $B$. If $B$ is torsion-free, then the word problem suffices because in this case $\ord(b) < \infty$ if, and only if, $b=_B1$. If $A$ is abelian, then the conditions (ii a) and (ii b) are equivalent, \ie we are always in case (ii a) and there is no need for a check whether $\ord(b) < \infty$.
		To be more explicit, we can write down the circuit as a formula (for the general non-abelian case): 
		\begin{align*}
		(b,f) \sim (c,g) \iff b\sim_B c &\land \bigl(\text{(i)} \lor \text{(ii)}\bigr).
		\end{align*}
		Moreover, we have
			\begin{align*}
			\text{(i)} &\iff \bigwedge_{i=1}^n \pi_{b_i,b}(f) =_A 1  \land  \bigwedge_{j=1}^m \pi_{\beta_j,c}(g) =_A 1,\\
			\text{(iii)} &\iff  \bigvee_{i=1}^n \left( \pi_{b_i,b}(f) \neq_A 1  \land \bigvee_{k=1}^m \left( \beta_k^{-1}b_ib=_Bc\beta_k^{-1}b_i \vphantom{\bigwedge_k \pi_{b_j}(f) } \right.\right.
			\\
			&\quad\quad\land\left.\left. \left[\left(\ord(b)=\infty \land \bigwedge_{t\in \wt T} \pi_{t,b}(f) =_A \pi_{t,b}^{(\beta_k^{-1}b_i)}(g)\right) \right.\right.\right.\\ 
			&\quad\qquad \lor \left.\left.\left.\left( \ord(b)<\infty \land \bigwedge_{t\in \wt T} \pi_{t,b}(f) \sim_A \pi_{t,b}^{(\beta_k^{-1}b_i)}(g)
			\right) \right]\right)\right).
			\end{align*}
		\end{proof}
	
	\begin{corollary}\label{cor:cpwreath}
		Let $A$ and $B$ be finitely generated groups and $d \geq 1$. Then
		\begin{itemize}
			\item  $\CP(A \wr^d B) \in \TC(\CP(A),\CP(B),\PP(B))$,
			\item  $\CP(A \wr^d B) \in \TC(\CP(A),\CP(B),\CMMP(B))$ if $B$ is torsion-free.
		\end{itemize}
	\end{corollary}
	\begin{proof}
		Immediate consequence of \prettyref{thm:cpwreath} by induction.
		\end{proof}
	Notice that $A \wr^d B$ is not abelian (for non-trivial $A$ and $B$). Hence, it does \emph{not} follow that $\CP(A \wr^d B) \in \TC(\CP(A),\CP(B),\CSMP(B))$ even if $A$ is abelian.
	
	The following quite trivial observation turns out to be very useful.
	
	\begin{lemma}\label{lem:distorsion}
		Let $G$ be finitely generated by $\Sig$ and let the order of its torsion elements be uniformly bounded. Suppose there is a polynomial $p(n)$ such that for every $w \in \Sig^*$ which is non-torsion, the inequality $k \leq p(\Abs{w^k})$ is satisfied, where $\Abs{w^k}$ denotes the geodesic length of the group element $w^k$. Then $\PP(G) \in \Ac0(\WP(G))$.
	\end{lemma}

	\begin{proof}
		Let $D$ be a bound on the order of torsion elements of $G$. For input words $v,w \in \Sig^*$ for the power problem, simply test whether $v^k=_G w$ for all $k$ with $ - p(\abs{w}) \leq k \leq \max\oneset{p(\abs{w}), D}$ in parallel using the word problem of $G$.
		\end{proof}
	
	The second condition of \prettyref{lem:distorsion} means that there is a uniform polynomial bound on the distortion of infinite cyclic subgroups. This is satisfied by abelian groups (with $p$ being linear). 
	Since the conjugacy problem in abelian groups is in \TC (as it is the word problem), we obtain the following corollary of \prettyref{thm:cpwreath}.
	
	\begin{corollary}
		Let $A$ and $B$ be \fg\ abelian groups and $d \geq 1$. Then  $\CP(A \wr^d B) \in \TC$.
	\end{corollary}

	\paragraph{The role of the power problem.} 
	The following result is a complexity analog of the ``only if'' part of \cite[Thm.\ B]{Matthews66}, which only considers decidability. Note that for pure decidability, it does not matter if we consider $\CSMP(B)$, $\CMMP(B)$ or $\PP(B)$ since they can all be reduced to each other. 
	
	\begin{theorem}\label{thm:lowerBound}
		Let $A$ be \fg\ and  non-trivial. Then $\CSMP(B) \leq_{\mathrm{m}}^{\Ac0} \CP(A \wr B)$. If, moreover, $A$ is non-abelian, then $\CMMP(B) \leq_{\mathrm{m}}^{\Ac0} \CP(A \wr B)$.
	\end{theorem}

	Notice that \prettyref{thm:lowerBound} shows that in the cases that $A$ is abelian or $B$ torsion-free \prettyref{thm:cpwreath} is the best possible result one could expect. However, it is totally unclear how $\PP(B)$ could be reduced to $\CP(A \wr B)$ in \TC (even if the answer to the power problem is guaranteed to have polynomial size). Thus, there remains the possibility that \prettyref{thm:cpwreath} could be strengthened in the general case.
	\begin{proof}
		The first statement is simply due to the observation that the construction in \cite[Thm.\ B]{Matthews66} can be computed in \Ac0. We repeat the argument here: fix some $a \in \Sig_A^*$ with $a\neq_A 1$. For $b,c \in \Sig_B^*$, the function $f \in A^{(B)}$ is defined by
		\begin{align*}
		f(1) &= a, &f(c) &= a^{-1}, & f(\beta) &= 1 \quad\text{for } \beta \in B \setminus\oneset{1,c}.
		\end{align*}
		Then by \prettyref{prop:conjugacyCrit} (ii), $(b,1) \sim (b,f)$ if, and only if, $\pi_{1,b}(f) = \pi_{c,b}(f) = 1$, which is the case if, and only if, $c \in \gen{b}$. Obviously, the tuples $(b,1)$ and $(b,f)$ can be computed in \Ac0.

		Now, let $A$ be non-abelian. In particular, there are elements $a_1, a_2 \in A$ with $a_1 a_2 \neq_A a_2 a_1$. For $b,c \in \Sig_B^*$, we define two functions $f,g \in A^{(B)}$ by
		\begin{align*}
		f(1) &= a_1a_2, & & & f(\beta) &= 1 \quad\text{for } \beta \in B \setminus\oneset{1},\\
		g(1) &= a_1, &g(c) &= a_2, & g(\beta) &= 1 \quad\text{for } \beta \in B \setminus\oneset{1,c}.
		\end{align*}
		Note that in the case $c=1$, technically $g$ is not well-defined; however, the group element $a_1 a_2^c$ is a valid input which can be written down (and in this case $g(1) = g(c)= a_1a_2$), so the reduction is still defined.
		
		We have $\pi_{1,b}(f) = a_1a_2$ and  $\pi_{t,b}(f) = 1$ for $t \not \in \gen{b}$. For $g$, according to \prettyref{prop:conjugacyCrit} (iii), we have to consider 	$\pi_{1,b}^{(1)}(g)$ and $\pi_{1,b}^{(c)}(g)$.
		If $b$ has finite order, then 
		$\pi_{1,b}^{(1)}(g)$ and $\pi_{1,b}^{(c)}(g)$ are both one of $a_1a_2$ or $a_2 a_1$ (which are conjugate) if, and only if, $c \in \gen{b} =_G \oneset{b}^*$ (because $b$ has finite order)~-- otherwise $\pi_{1,b}^{(1)}(g)= a_1$ and $\pi_{1,b}^{(c)}(g)=a_2$. 
		On the other hand if $b$ has infinite order, we have \begin{align*}
		\pi_{1,b}^{(1)}(g)  &= \begin{cases}
		a_1 a_2 & \text{if } c=_B b^k \text{ with } k \geq 0,\\
		a_2 a_1 & \text{if } c=_B b^{k} \text{ with } k < 0,\\
		a_1 &  \text{otherwise},
		\end{cases} &
		\pi_{1,b}^{(c)}(g) &= \begin{cases}
		a_1 a_2 & \text{if } c=_B b^k \text{ with } k \geq 0,\\
		a_2 a_1 & \text{if } c=_B b^{k} \text{ with } k < 0,\\
		a_2 &  \text{otherwise}.
		\end{cases}
		\end{align*}
		Thus, $\pi_{1,b}^{(d)}(g)= \pi_{1,b}(f)$ for some $d \in \oneset{1,c}$ if, and only if, $c=_B b^k$ with $k \geq 0$.
		Therefore, by \prettyref{prop:conjugacyCrit}, $(b,f) \sim (b,g)$ if, and only if, $c \in_G \oneset{b}^*$.
		\end{proof}

	\section{Conjugacy and Power Problem in Left-Iterated Wreath Products} \label{sec:itwreath}
	
	In order to solve the conjugacy problem in left-iterated wreath products, we also need to solve the power problem in wreath products. In general, we do not know whether the power problem in a wreath product is in \TC given that the power problem of the factors is in \TC. The issue is that when dealing with torsion it might be necessary to compute greatest common divisors~-- which is not known to be in \TC. By restricting torsion elements to have only smooth orders, we circumvent this issue. Recall that a number is called $\beta$-smooth for some $\beta \in \N$ if it only contains prime factors less than or equal to $\beta$. 
	
	\begin{lemma}\label{lem:boundedorder}
		Let $\beta \in \N$. Suppose the orders of all torsion elements in $A$ and $B$ are $\beta$-smooth. Then the orders of all torsion elements in $A \wr B$ are $\beta$-smooth.
	\end{lemma}
	
	\begin{proof}
		First consider a torsion element $f \in A^{(B)}$. Since $f(b)$ is a torsion element for all $b \in B$, we have $f^\ell=1$ for some $\beta$-smooth $\ell$.
		Next consider some arbitrary torsion element $(b,f) \in A \wr B$. Then $b$ is torsion as well (since $(b,f)$ projects to $b$ in $B$), and thus $b^m= 1$ for some $\beta$-smooth $m$. Consequently, $(b,f)^m \in A^{(B)}$. Hence, $(b,f)^{\ell m} = 1$ and so the order of $(b,f)$ is $\beta$-smooth.	
		\end{proof}
	
	Note that we are not aware of any finitely generated group with word problem in \TC and torsion elements whose orders are not $\beta$-smooth for any $\beta$. On the other hand, there are recursively presented such groups: for instance, take the infinite direct sum of cyclic groups of arbitrary order.

	We say $(t_1, \dots, t_m)$ is a list of \emph{$\gen{b}$-coset representatives} if the $t_i$ represent pairwise distinct $\gen{b}$-cosets.
	\begin{lemma}\label{lem:supportdecomposition}
		The following problems are in $\TC(\PP(B))$:
		\begin{enumerate}
			\item Input: a function $f=(( b_1,  a_1), \ldots, ( b_n,  a_n)) \in A^{(B)}$ and $b \in \Sig_B^*$.\\
			Output: a list of $\gen{b}$-coset representatives $(t_1, \dots, t_m)$ such that $\supp(f) \sse \oneset{t_1, \dots, t_m}\!\cdot\!\gen{b}$.\label{asdf}
			
			\item Input: a function $f=(( b_1,  a_1), \ldots, ( b_n,  a_n)) \in A^{(B)}$, $b \in \Sig_B^*$ and a list of $\gen{b}$-coset representatives $(t_1, \dots, t_m)$.\\
			Decide whether $\supp(f) \sse \oneset{t_1, \dots, t_m}\cdot\gen{b}$.\label{asdfg}

			\item\label{asdfgh} Input: a function $f=(( b_1,  a_1), \ldots, ( b_n,  a_n)) \in A^{(B)}$, $b \in \Sig_B^*$ and a list of $\gen{b}$-coset representatives $(t_1, \dots, t_m)$.\\
			Output: for each  $1 \leq i \leq m$ a list $((e_{i,1},a_{i,1}), \dots, (e_{i,n_i},a_{i,n_i}))$ with $e_{i,j} \in \Z$ (encoded in binary), $e_{i,1} < \cdots < e_{i,n_i}$ and $a_{i,j} \in \Sig_A^*$  such that
			\begin{align*}
			\supp(f) &= \set{t_ib^{e_{i,j}}}{1\leq i \leq m, 1\leq j \leq n_i} \quad\text{ and }&f(t_jb^{e_{i,j}}) &= a_{i,j}.
			\end{align*}	
		\end{enumerate}
	\end{lemma}
	We assume that all the words in the output of the circuit of \prettyref{lem:supportdecomposition} are encoded with the same number of bits (the number of bits is a fixed polynomial in the number of input bits depending only on the group $B$). 
	\begin{proof}
		\ref{asdf}	The first layer of the circuit decides for all $i,j \leq n$ in parallel whether $b_i^{-1}b_j \in \gen{b}$ using oracle gates for the power problem in $B$. 
		In the next layer, an element $b_i$ is included in the list of representatives $(t_1, \dots, t_m)$ if, and only if, there is no $j<i$ with $b_i^{-1}b_j \in \gen{b}$.
		
		\ref{asdfg} One simply needs to check whether for all $j\in\oneset{1, \dots n}$ there is some $i \in\oneset{1, \dots m}$ such that $t_i^{-1}b_j \in \gen{b}$  using the power problem of $B$.
		
		\ref{asdfgh} For all $i \in\oneset{1, \dots m}$ and $j\in\oneset{1, \dots n}$ one checks whether $t_i^{-1}b_j \in \gen{b}$ and, if so, computes the respective exponent $e_{i,j}$ such that $t_i^{-1}b_j= b^{e_{i,j}}$. For all $i$ and $j$ this can be done in parallel by using oracle gates to the power problem for $B$. 
		
		The next step is to sort for all $i$ in parallel the tuples $(b_j, a_j)$ with $t_i^{-1}b_j \in \gen{b}$ according to their exponent $e_{i,j}$. This can be done in \TC as described in \prettyref{ex:sortTC}. This yields the output lists.
		\end{proof}
	For the proof of \prettyref{thm:pp}, we need some more notation: 
	for $k >0$, $b \in B$, and $f \in A^{(B)}$, we define $f^{(b, k)}$ by $(b,f)^k = (b^k,f^{(b, k)})$. Then we have
	\begin{align}
	f^{(b, k)}(c) &= (f^{b^{k-1}}\cdots f^b f)(c)\qquad \text{for } c \in B. \label{eq:fbc}
	\end{align}
	\begin{lemma}\label{lem:evalF}
		Let $e_1, \dots, e_n \in \Z$ with $e_1 < \dots < e_n$ and $a_1, \dots, a_n \in A$ and $b,t \in B$.  Furthermore, let $f(tb^{e_i}) = a_i$ for $i = 1, \dots, n$ and $f(c) = 1$ for all other $c \in B$.		
		Then, for $0 < k \leq \ord(b)$, we have\begin{align*}
		f^{(b, k)}(t b^\ell) &= a_i \cdots a_{j-1}
		\end{align*}
		for $1 \leq i \leq j \leq n+1$ such that $\max\oneset{e_{j-1}, e_{i-1} + k} \leq \ell \leq \min\oneset{e_i + k-1, e_{j} - 1}$.  Here, we set $e_0 = -\infty$ and $e_{n+1} = \infty$.  Note that $a_i \cdots a_{j-1}$ is possibly the empty product. 
	\end{lemma}
	\begin{proof}
		\begin{align*}
		f^{(b, k)}(t b^\ell) &= (f^{b^{k-1}}\cdots f^b f)(tb^\ell)\\
		&= f(t b^{\ell-(k-1)}) \cdots f(t b^{\ell-1}) f(tb^\ell)\\
		&= \prod_{\nu= i}^{j-1}f(t b^{e_\nu}) \tag{because all other $f(c)$ are trivial}
		\end{align*}
		for $i = \min\set{\nu}{e_\nu \geq \ell-(k-1)}$ and $j = \max \set{\nu}{e_\nu \leq \ell} + 1$.
		Thus, $e_{i-1} < \ell-(k-1)  \leq  e_i$ and likewise  $e_{j-1}\leq \ell \leq  e_{j} - 1$.
		\end{proof}
	
	\begin{lemma}\label{lem:computeFbk}
		The following problem is in $\TC(\PP(B))$:\\
		Input: a function $f=(( b_1,  a_1), \ldots, ( b_n,  a_n)) \in A^{(B)}$, $b,t \in \Sig_B^*$ such that $\supp(f) \sse t\gen{b}$ and $k,\ell \in \Z$ (in binary).\\
		Compute $f^{(b,k)}(tb^\ell) \in \Sig_A^*$.	
		
	\end{lemma}
	
	\begin{proof}
		By \prettyref{lem:supportdecomposition}, we can compute a representation $((e_1,a_1), \dots, (e_n,a_n))$ with $e_1 < \cdots < e_n$ of $f$ such that $f(tb^{e_i}) = a_i$ for all $i$ and $f(c) = 1$, otherwise.	
		
		By \prettyref{lem:evalF}, $f^{(b, k)}(tb^{e_\nu})$ is of the form $a_i \cdots a_{j-1}$ for appropriate $i$ and $j$. The indices $i$ and $j$ can be found by evaluating the inequality  $\max\oneset{e_{j-1}, e_{i-1} + k} \leq e_\nu \leq \min\oneset{e_i + k-1, e_{j} - 1}$~-- that is a simple Boolean combination of comparisons of integers (integers can be compared in \TC \eg by subtracting them and then checking the sign). 
		\end{proof}
	
	\begin{theorem}\label{thm:pp}
		Let $\beta \in \N$ and suppose the order of every torsion element in $A$ is $\beta$-smooth. Then we have $\PP(A \wr B) \in \TC(\PP(A), \PP(B))$.
	\end{theorem}

	Roughly the  proof of \prettyref{thm:pp} works as follows: on input $(b,f)$ and $(c,g)$ first apply the power problem in $B$ to $b$ and $c$. If there is no solution, then there is also no solution for $(b,f)$ and $(c,g)$. Otherwise, the smallest $k\geq 0$ with $b^k=_Bc$ can be computed. If $b$ has infinite order, it remains to check whether $(b,f)^k = (c,g)$. Since $k$ might be too large, this cannot be done by simply applying the word problem. Nevertheless, we only need to establish equality of functions in $A^{(B)}$. We show that it suffices to check equality on certain (polynomially many) ``test points''.
	In the case that $b$ has finite order $K$, we know that if there is a solution to the power problem it must be in $k + K\Z$. Now, similar techniques as in the infinite order case can be applied to find the solution.

	\begin{proof}
		By \prettyref{lem:reduce}, we may assume that the input is given as two pairs $(b,f)$ and $(c,g)$. We aim to compute some $k$ such that  $(b,f)^k = (c,g)$ if there exists such $k$. We describe a circuit in several stages. It will use oracle gates for $\PP(A)$, $\PP(B)$ as well as sorting in \TC and integer arithmetic. As in the previous proofs it is straightforward to assign gate numbers such that on input of two gate numbers it can be decided in linear time whether there is a wire connecting them. 
		As a first step, the power problem in $B$ is applied to determine whether there is some $k$ with $b^k= c $. If there the answer is ``no'', then the over all answer is ``no''. Otherwise, we distinguish the two cases that $b$ is of finite order and that $b$ is of infinite order (which can be distinguished by using the power problem). 		
		
		First assume that $b$ has infinite order. Let $k$ be the answer for the power problem in $b$ and $c$, \ie $k$ is the unique integer with $b^{k}=_B c$.  Now, it remains to check whether $(b,f)^k = (c,g)$. This cannot be done by simply applying the word problem because $k$ might be exponentially large (we know that it is bounded by some exponential function because it can be computed in \TC). Without loss of generality, we may assume that $k > 0$. Indeed, if $k<0$, we can replace $(c,g)$ by $(c,g)^{-1}$ and, if $k=0$, we only need to check whether $g=0$ in order to establish $(b,f)^k = (c,g)$.

		Since $k>0$, by \prettyref{eq:fbc}, we have $(b,f)^k =  (b^k,f^{(b,k)}) $ where $f^{(b,k)} = f^{b^{k-1}}\cdots f^b f$~-- thus, we have to compare $f^{(b,k)}$ and $g$ for equality in $A^{(B)}$.
		By \prettyref{lem:supportdecomposition} \ref{asdf}, a list of $\gen{b}$-coset representatives $(t_1, \dots, t_m)$ can be computed in $\TC(\PP(B))$ such that $\supp(f) \sse \oneset{t_1, \dots, t_m} \cdot \gen{b}$. Because of \prettyref{eq:fbc}, also $\supp(f^{(b,k)}) \sse \oneset{t_1, \dots, t_m} \cdot \gen{b}$. Thus, if $\supp(g) \not\sse \oneset{t_1, \dots, t_m} \cdot \gen{b}$ (which can be checked in $\TC(\PP(B))$ by \prettyref{lem:supportdecomposition} \ref{asdfg}), then $f^{(b,k)} \neq g$.
		
		Because $f^{(b, k)}  = g$ if, and only if, they agree on every $b$-coset, we can assume that $\supp(f), \supp(g) \sse t\gen{b}$ for some $t \in B$~-- the general case is then simply a conjunction over all coset representatives. 
		By \prettyref{lem:supportdecomposition} \ref{asdfgh}, we can compute representations $((e_1,a_1), \dots, (e_n,a_n))$ with $e_1 < \cdots < e_n$ (resp.\ $((e'_1,a'_1), \dots, (e'_n,a'_{n'}))$ with $e'_1 < \cdots < e'_{n'}$) of $f$ (resp.\ $g$) such that $f(tb^{e_i}) = a_i$ for all $i$ and $f(c) = 1$, otherwise (and likewise for $g$).

		\prettyref{lem:evalF} allows us to compare $f^{(b, k)}$ and $g$ for equality. We do this in two steps: first we check for all $tb^{e'_\nu} \in \supp(g)$ (\ie for $\nu = 1, \dots, n'$) whether $f^{(b, k)}(tb^{e'_\nu})  =_A g(tb^{e'_\nu})$. We can find $f^{(b, k)}(tb^{e'_\nu})$ by \prettyref{lem:computeFbk}. Now it remains to check whether $f^{(b, k)}(tb^{e'_\nu}) =_A a'_\nu= g(tb^{e'_\nu})$ using oracle gates for the word problem of $A$.  For all $\nu$ this can be done in parallel. 
		
		At this point, we know that $f^{(b, k)}$ and $g$ agree on $\supp(g)$. The second step is to check that $\supp(f^{(b, k)}) \sse \supp(g)$. Since $\supp(f^{(b, k)})$ might be exponentially large, we have to use a different strategy than a point-wise check. Instead, we do the following for all $1 \leq i \leq j \leq n+1$ in parallel:

		\begin{itemize}
			\item Check whether $a_i \cdots a_{j-1}=_A 1$ (can be checked with oracle gates for $\WP(A)$). If not, then there are two possibilities:
			
			\item  If $\min\oneset{e_i + k-1, e_{j} - 1} - \max\oneset{e_{j-1}, e_{i-1} + k} > n$, then by \prettyref{lem:evalF} $\abs{\supp(f^{(b, k)})} > n \geq \abs{\supp(g)}$; thus, we know that $f^{(b, k)} \neq g$.
			\item Otherwise, test for all $\ell$ satisfying $\max\oneset{e_{j-1}, e_{i-1} + k} \leq \ell\leq \min\oneset{e_i + k-1, e_{j} - 1}$ whether there is some $\nu$ with  $\ell = e'_{\nu}$ (since it is a simple disjunction over equality tests of integers, it can be done in \TC). If there is some $\ell$ which is not equal to any $e'_{\nu}$, then $\supp(f^{(b, k)}) \not\sse \supp(g)$. 
			
		\end{itemize}
		If none of the above cases refutes that $f^{(b,k)} = g$, then we know that indeed $f^{(b,k)} = g$.\bigskip		
		
		\newcommand{\ordd}{K}
		
		Now, let $b$ have finite order $\ordd$ and let $0\leq k < \ordd$ with $b^{k} = c$~-- \ie $k$ is the solution to the power problem for $b$ and $c$. As remarked before, also $\ordd$ can be computed by using the power problem for $B$. We have $(b,f)^\ordd \in A^{(B)}$. Moreover, $b^{k'} = c$ for $k' \in \Z$ if, and only if, $k' \equiv k \mod \ordd$.
		Thus, there is a solution to the power problem if, and only if, there is some $\ell \in \Z$ with $((b,f)^\ordd)^\ell (b,f)^{k} = (c,g)$. In other words it remains to solve the power problem for $(b,f)^\ordd$ and $(c,g)(b,f)^{-k}$. We can simplify the latter element as follows
		\begin{align*}\allowdisplaybreaks
		(c,g)(b,f)^{-k} &= (c,g)\left((b,f)^{k}\right)^{-1} = (c,g)\left(b^{k},f^{(b,k)}\right)^{-1}  \\ &= (c,g)\left(b^{-k},{\left((f^{(b,k)})^{b^{-k}}\right)}^{-1}\right)  
		\\ &= \left(cb^{-k},g^{b^{-k}}\cdot {\left((f^{(b,k)})^{b^{-k}}\right)}^{-1}\right)= \left(1,{\left(g \cdot (f^{(b,k)})^{-1}\right)}^{b^{-k}}\right)
		\end{align*}
		and we see that we have to solve the power problem for $f^{(b,\ordd)}$ and $g^{b^{-k}} \cdot {\left((f^{(b,k)})^{b^{-k}}\right)}^{-1}$ in $A^{(B)}$.
		Note that since the numbers $k$ and $\ordd$ might be exponential in the input size, these group elements cannot be written down completely inside the polynomial size circuit.

		We start as in the infinite order case:
		by \prettyref{lem:supportdecomposition} \ref{asdf}, a list of $\gen{b}$-coset representatives $(t_1, \dots, t_m)$ with $\supp(f) \sse \oneset{t_1, \dots, t_m} \cdot \gen{b}$ can be computed in \TC. Because of \prettyref{eq:fbc}, also $\supp(f^{(b,\ordd)}) \sse \oneset{t_1, \dots, t_m} \cdot \gen{b}$~-- thus, again, if $\supp(g) \not\sse \oneset{t_1, \dots, t_m} \cdot \gen{b}$ (which can be checked in $\TC(\PP(B))$ by \prettyref{lem:supportdecomposition} \ref{asdfg}), then we already know that $((b,f)^\ordd)^\ell (b,f)^{k} \neq (c,g)$ for any $\ell$.
		
		In the following, we assume again that $\supp(f), \supp(g) \sse t\gen{b}$ for some $t \in B$~-- the set of solutions to the general case is the intersection over the solution sets for all coset representatives. In the end we will show how to compute this intersection.
		By \prettyref{lem:supportdecomposition} \ref{asdfgh}, we can compute representations $((e_1,a_1), \dots, (e_n,a_n))$ with $e_1 < \cdots < e_n$ (resp.\ $((e'_1,a'_1), \dots, (e'_n,a'_{n'}))$ with $e'_1 < \cdots < e'_{n'}$) of $f$ (resp.\ $g$) such that $f(tb^{e_i}) = a_i$ and $f(c) = 1$, otherwise (and likewise for $g$).

		We have to solve the power problem for $f^{(b,\ordd)}(tb^\ell)$ and $g(tb^{\ell + k})\cdot f^{(b,k)}(tb^{\ell + k})^{-1}$ for all $\ell$. Since again there might be too many points in the support of $f^{(b,\ordd)}$, we have to restrict to certain test points.
		
		By \prettyref{lem:evalF}, for each $\gen{b}$-coset intersecting $\supp(f)$ there are lists $\gamma_0, \dots, \gamma_\nu \in \Z$ and $\alpha_1, \dots, \alpha_\nu \in \Sig_A^*$ with $f^{(b,\ordd)}(tb^\ell) = \alpha_i$ for all $\gamma_{i-1} < \ell \leq \gamma_i$ and $\gamma_0 + \ordd = \gamma_\nu$ (with $\nu \leq 2n + 1$). The numbers $\gamma_i$ can be computed in \TC like in \prettyref{lem:computeFbk}. Moreover, we can compute similar lists $\gamma'_0, \dots, \gamma'_{\nu'} \in \Z$, $\alpha'_1, \dots, \alpha'_{\nu'} \in \Sig_A^*$ for $f^{(b,k)}$ and $\gamma''_0, \dots, \gamma''_{\nu''} \in \Z$, $\alpha''_1, \dots, \alpha''_{\nu''} \in \Sig_A^*$ for $g$. Now, it suffices to solve the power problem for $ f^{(b,\ordd)}(tb^\ell) = \alpha_{i_\ell} $ (where $i_\ell$ is such that $\gamma_{i_\ell-1} < \ell \leq \gamma_{i\ell}$) and $(g(tb^\ell)\cdot f^{(b,k)}(tb^\ell)^{-1})^{b^{-k}}$ for all 
		$$\ell \in \oneset{\gamma_0, \dots, \gamma_\nu, \gamma'_0 - k, \dots, \gamma'_{\nu'}- k, \gamma''_0 - k, \dots, \gamma''_{\nu''} - k}=:\Gamma.$$
		This is because for increasing $\ell$, the values $f^{(b,k)}(tb^\ell)$, $g(tb^\ell)^{b^{-k}}$, and $f^{(b,K)}(tb^\ell)^{b^{-k}}$ only change at these points. The functions can be evaluated in $\TC(\PP(B))$ by \prettyref{lem:computeFbk}, then oracle gates for $\PP(B)$ are used.
		
		For $\ell \in \Gamma$, let $K_\ell$ denote the order of $\alpha_{i_\ell}$ and let $k_\ell\in \Z$ such that $\alpha_{i_\ell} ^{k_\ell} = (g(tb^\ell)\cdot f^{(b,k)}(tb^\ell)^{-1})^{b^{-k}}$ (\ie the solution to the power problem). 
		We obtain a system of congruences
		$$x \equiv k_\ell \mod \ordd_\ell$$
		(here congruent modulo $\infty$ means equality). Since the $\ordd_\ell$ are all $\beta$-smooth, they can be factored in \TC and a solution (if there is one) of this system can be determined in \TC (see \eg \cite[Lem.\ 27]{Weiss16}) with the help of Hesse's division circuit \cite{hesse01,HeAlBa02} using the Chinese remainder theorem. 
		
		We do this also for all coset representatives in parallel. In the end, we either see that $g^{b^{-k}} {\left((f^{(b,k)})^{b^{-k}}\right)}^{-1}$ is not a power of $f^{(b,\ordd)}$, or we obtain a list of solutions $x_1, \dots, x_{m} \in \Z$, which give rise to a system of congruences which can be solved like in the preceding paragraph.
		\end{proof}

	By repeated application of \prettyref{thm:cpwreath}, \prettyref{lem:boundedorder}, and \prettyref{thm:pp}, we obtain the first statement of \prettyref{cor:itwrCP} below. The second statement follows since the Magnus embedding preserves conjugacy \cite{RemSok} (that means two elements are conjugate in the free solvable group if, and only if, their images under the Magnus embedding are conjugate). 
	
	\begin{corollary}\label{cor:itwrCP}
		Let $A$ and $B$ be f.\,g.\ abelian groups and let $d \geq 1$. The conjugacy problem of $A\mathop{^d\wr}B$ is in \TC. Also, the conjugacy problem of free solvable groups is in \TC.
	\end{corollary}
	\begin{remark}
		In \prettyref{cor:itwrCP}, $A$ and $B$ can also be chosen to be a solvable Baumslag-Solitar group \BS1q since the power problem is in \TC by \prettyref{ex:bs}, the conjugacy problem is in \TC by \cite{DiekertMW14}, and these groups are torsion-free. 
		
		Moreover, in \cite{MyasnikovW17} it is shown that also nilpotent groups have power problem and conjugacy problem in \TC and that the orders of torsion elements are uniformly bounded. Thus, also iterated wreath products of nilpotent groups have conjugacy problem in \TC.
	\end{remark}
	
	\section{Conclusion and Open Problem}\label{sec:conclusion}
	As already discussed in \prettyref{qu:majoritydepth}, an important open problem is the dependency of the depth of the circuits for the word problem on the solvability degree.
	
	We have seen how to solve the conjugacy problem in a wreath product in \TC with oracle calls to the conjugacy problems of both factors and the power problem (resp.\ cyclic submonoid/subgroup membership problem) in the second factor. However, we do not have a reduction from the power problem in the second factor to the conjugacy problem in the wreath product: even if $A$ is non-abelian, we only know that the cyclic submonoid membership problem is necessary to solve the conjugacy problem in the wreath product.
	\begin{question}
		Is $\CP(A \wr B) \in \TC(\CP(A),\CP(B),\CMMP(B))$ in general? 
	\end{question}
	
	For iterated wreath products we needed the power problem to be in \TC in order to show that the conjugacy problem is in \TC. One reason was that we only could reduce the power problem in the wreath product to the power problems of the factors. However, we have seen that in torsion-free groups, we do not need the power problem to solve conjugacy, as the cyclic submonoid membership problem is sufficient. Therefore, it would be interesting to reduce the cyclic submonoid membership problem in a wreath product to the same problem in its factors.
	\begin{question}
		Is $\CMMP(A \wr B) \in \TC(\CMMP(A),\CMMP(B))$ or similarly is $\CSMP(A \wr B) \in \TC(\CSMP(A),\CSMP(B))$? 
	\end{question}

	In \cite{GulSU17}, Gul, Sohrabi, and Ushakov generalized Matthews result by considering the relation between the conjugacy problem in $F/N$ and the power problem in $F/N'$, where $F$ is a free group with a normal subgroup $N$ and $N'$ is its derived subgroup. They show that $\CP(F/N')$ is polynomial-time-Turing-reducible to $\CSMP(F/N)$ and $\CSMP(F/N)$ is Turing-reducible to $\CP(F/N')$ (no complexity bound). Moreover, they establish that $\WP(F/N')$ is polynomial-time-Turing-reducible to $\WP(F/N)$.
	
	\begin{question}
		What are the precise relations in terms of complexity between
		$\CP(F/N')$ and  $\CSMP(F/N)$ resp.\ $\WP(F/N')$ and $\WP(F/N)$?
	\end{question}


\begin{thebibliography}{10}
		\providecommand{\url}[1]{{#1}}
		\providecommand{\urlprefix}{URL }
		\expandafter\ifx\csname urlstyle\endcsname\relax
		\providecommand{\doi}[1]{DOI~\discretionary{}{}{}#1}\else
		\providecommand{\doi}{DOI~\discretionary{}{}{}\begingroup
			\urlstyle{rm}\Url}\fi
		
		\bibitem{BarringtonIS90}
		Barrington, D.A.M., Immerman, N., Straubing, H.: On uniformity within
		{NC\({^1}\)}.
		\newblock J. Comput. Syst. Sci. \textbf{41}(3), 274--306 (1990).
		\newblock \doi{10.1016/0022-0000(90)90022-D}.
		\newblock \urlprefix\url{http://dx.doi.org/10.1016/0022-0000(90)90022-D}
		
		\bibitem{CravenJ12}
		Craven, M.J., Jimbo, H.C.: Evolutionary algorithm solution of the multiple
		conjugacy search problem in groups, and its applications to cryptography.
		\newblock Groups Complexity Cryptology \textbf{4}, 135--165 (2012)
		
		\bibitem{Dehn11}
		Dehn, M.: \"{U}ber unendliche diskontinuierliche {G}ruppen.
		\newblock Math. Ann. \textbf{71}(1), 116--144 (1911).
		\newblock \doi{10.1007/BF01456932}.
		\newblock \urlprefix\url{http://dx.doi.org/10.1007/BF01456932}
		
		\bibitem{DiekertMW14}
		Diekert, V., Myasnikov, A.G., Wei{\ss}, A.: Conjugacy in {B}aumslag's {G}roup,
		{G}eneric {C}ase {C}omplexity, and {D}ivision in {P}ower {C}ircuits.
		\newblock In: Latin American Theoretical Informatics Symposium, pp. 1--12
		(2014)
		
		\bibitem{GrigorievS09}
		Grigoriev, D., Shpilrain, V.: Authentication from matrix conjugation.
		\newblock Groups Complexity Cryptology \textbf{1}, 199--205 (2009)
		
		\bibitem{GulSU17}
		Gul, F., Sohrabi, M., Ushakov, A.: Magnus embedding and algorithmic properties
		of groups {$F/N^{(d)}$}.
		\newblock Trans. Amer. Math. Soc. \textbf{369}(9), 6189--6206 (2017).
		\newblock \doi{10.1090/tran/6880}.
		\newblock \urlprefix\url{http://dx.doi.org/10.1090/tran/6880}
		
		\bibitem{hesse01}
		Hesse, W.: Division is in uniform {TC}$^{0}$.
		\newblock In: F.~Orejas, P.G. Spirakis, J.~van Leeuwen (eds.) ICALP,
		\emph{Lecture Notes in Computer Science}, vol. 2076, pp. 104--114. Springer
		(2001)
		
		\bibitem{HeAlBa02}
		Hesse, W., Allender, E., Barrington, D.A.M.: Uniform constant-depth threshold
		circuits for division and iterated multiplication.
		\newblock JCSS \textbf{65}, 695--716 (2002)
		
		\bibitem{KargapolovR66}
		Kargapolov, M.I., Remeslennikov, V.N.: The conjugacy problem for free solvable
		groups.
		\newblock Algebra i Logika Sem. \textbf{5}(6), 15--25 (1966)
		
		\bibitem{KoLCHKP00}
		Ko, K.H., Lee, S.J., Cheon, J.H., Han, J.W., Kang, J.s., Park, C.: New
		public-key cryptosystem using braid groups.
		\newblock In: Advances in cryptology---{CRYPTO} 2000 ({S}anta {B}arbara, {CA}),
		\emph{Lecture Notes in Comput. Sci.}, vol. 1880, pp. 166--183. Springer,
		Berlin (2000).
		\newblock \doi{10.1007/3-540-44598-6_10}.
		\newblock \urlprefix\url{http://dx.doi.org/10.1007/3-540-44598-6_10}
		
		\bibitem{KoenigL17}
		K{\"o}nig, D., Lohrey, M.: Evaluation of circuits over nilpotent and polycyclic
		groups.
		\newblock Algorithmica  (2017).
		\newblock \doi{10.1007/s00453-017-0343-z}.
		\newblock \urlprefix\url{https://doi.org/10.1007/s00453-017-0343-z}
		
		\bibitem{KrebsLR07}
		Krebs, A., Lange, K., Reifferscheid, S.: Characterizing {TC}$^{0}$ in terms of
		infinite groups.
		\newblock Theory Comput. Syst. \textbf{40}(4), 303--325 (2007).
		\newblock \doi{10.1007/s00224-006-1310-2}.
		\newblock \urlprefix\url{http://dx.doi.org/10.1007/s00224-006-1310-2}
		
		\bibitem{LangeM98}
		Lange, K., McKenzie, P.: On the complexity of free monoid morphisms.
		\newblock In: K.~Chwa, O.H. Ibarra (eds.) Algorithms and Computation, 9th
		International Symposium, {ISAAC} '98, Taejon, Korea, December 14-16, 1998,
		Proceedings, \emph{Lecture Notes in Computer Science}, vol. 1533, pp.
		247--256. Springer (1998).
		\newblock \doi{10.1007/3-540-49381-6_27}.
		\newblock \urlprefix\url{http://dx.doi.org/10.1007/3-540-49381-6_27}
		
		\bibitem{MacielT98}
		Maciel, A., Th{\'{e}}rien, D.: Threshold circuits of small majority-depth.
		\newblock Inf. Comput. \textbf{146}(1), 55--83 (1998).
		\newblock \doi{10.1006/inco.1998.2732}.
		\newblock \urlprefix\url{http://dx.doi.org/10.1006/inco.1998.2732}
		
		\bibitem{Magnus39}
		Magnus, W.: On a theorem of {M}arshall {H}all.
		\newblock Ann. of Math. (2) \textbf{40}, 764--768 (1939)
		
		\bibitem{Matthews66}
		Matthews, J.: The conjugacy problem in wreath products and free metabelian
		groups.
		\newblock Transaction of the American Math Society \textbf{121}, 329--339
		(1966)
		
		\bibitem{MiasnikovVW17}
		Miasnikov, A., Vassileva, S., Wei{\ss}, A.: The conjugacy problem in free
		solvable groups and wreath products of abelian groups is in $\mathsf{TC}^0$.
		\newblock In: P.~Weil (ed.) Computer Science - Theory and Applications - 12th
		International Computer Science Symposium in Russia, {CSR} 2017, Kazan,
		Russia, June 8-12, 2017, Proceedings, \emph{Lecture Notes in Computer
			Science}, vol. 10304, pp. 217--231. Springer (2017).
		\newblock \doi{10.1007/978-3-319-58747-9_20}.
		\newblock \urlprefix\url{https://doi.org/10.1007/978-3-319-58747-9_20}
		
		\bibitem{Miller1}
		{Miller III}, C.F.: {On group-theoretic decision problems and their
			classification}, \emph{Annals of Mathematics Studies}, vol.~68.
		\newblock Princeton University Press (1971)
		
		\bibitem{MyasnikovRUV10}
		Myasnikov, A., Roman'kov, V., Ushakov, A., Vershik, A.: The word and geodesic
		problems in free solvable groups.
		\newblock Trans. Amer. Math. Soc. \textbf{362}(9), 4655--4682 (2010).
		\newblock \doi{10.1090/s0002-9947-10-04959-7}.
		\newblock \urlprefix\url{http://dx.doi.org/10.1090/s0002-9947-10-04959-7}
		
		\bibitem{MyasnikovW17}
		{Myasnikov}, A., {Wei{\ss}}, A.: {TC$^0$ circuits for algorithmic problems in
			nilpotent groups}.
		\newblock ArXiv e-prints  (2017)
		
		\bibitem{MiasnikovVW16}
		Myasnikov, A.G., Vassileva, S., Wei{\ss}, A.: Log-space complexity of the
		conjugacy problem in wreath products. To appear
		
		\bibitem{RemSok}
		Remeslennikov, V., Sokolov, V.G.: Certain properties of the magnus embedding.
		\newblock Algebra i logika \textbf{9(5)}, 566--578 (1970)
		
		\bibitem{Robinson93phd}
		Robinson, D.: Parallel algorithms for group word problems.
		\newblock Ph.D. thesis, University of California, San Diego (1993)
		
		\bibitem{SZ1}
		{Shpilrain}, V., {Zapata}, G.: Combinatorial group theory and public key
		cryptography.
		\newblock Appl. Algebra Engrg. Comm. Comput. \textbf{17}, 291--302 (2006)
		
		\bibitem{Vassileva11}
		Vassileva, S.: Polynomial time conjugacy in wreath products and free solvable
		groups.
		\newblock Groups Complex. Cryptol. \textbf{3}(1), 105--120 (2011).
		\newblock \doi{10.1515/GCC.2011.005}.
		\newblock \urlprefix\url{http://dx.doi.org/10.1515/GCC.2011.005}
		
		\bibitem{Vollmer99}
		Vollmer, H.: Introduction to Circuit Complexity.
		\newblock Springer, Berlin (1999)
		
		\bibitem{Waack90}
		Waack, S.: The parallel complexity of some constructions in combinatorial group
		theory.
		\newblock In: Proceedings on Mathematical Foundations of Computer Science 1990,
		MFCS '90, pp. 492--498. Springer-Verlag New York, Inc., New York, NY, USA
		(1990).
		\newblock \urlprefix\url{http://dl.acm.org/citation.cfm?id=88581.90249}
		
		\bibitem{WangWCO11}
		Wang, L., Wang, L., Cao, Z., Okamoto, E., Shao, J.: New constructions of
		public-key encryption schemes from conjugacy search problems.
		\newblock In: Information security and cryptology, \emph{Lecture Notes in
			Comput. Sci.}, vol. 6584, pp. 1--17. Springer, Heidelberg (2011).
		\newblock \doi{10.1007/978-3-642-21518-6_1}.
		\newblock \urlprefix\url{http://dx.doi.org/10.1007/978-3-642-21518-6_1}
		
		\bibitem{Weiss16}
		Wei\ss, A.: A logspace solution to the word and conjugacy problem of
		generalized {B}aumslag-{S}olitar groups.
		\newblock In: Algebra and computer science, \emph{Contemp. Math.}, vol. 677,
		pp. 185--212. Amer. Math. Soc., Providence, RI (2016)
		
	\end{thebibliography}

\end{document}